\documentclass[pdflatex,sn-vancouver-num]{sn-jnl}
\usepackage[utf8]{inputenc}
\usepackage[T1]{fontenc}
\usepackage{amsmath,amssymb,amsthm,graphicx,booktabs}
\usepackage{xurl}
\usepackage{microtype}
\makeatletter
\renewcommand\normalsize{%
  \@setfontsize\normalsize{12pt}{15pt}%
  \abovedisplayskip 12pt plus 2pt minus 1pt
  \abovedisplayshortskip 0pt plus 3pt
  \belowdisplayshortskip 3pt plus 3pt minus 2pt
  \belowdisplayskip\abovedisplayskip
  \let\@listi\@listI}
\normalsize
\let\small\normalsize
\let\footnotesize\normalsize
\let\scriptsize\normalsize
\@ifundefined{addressfont}{}{\def\addressfont{\normalfont\normalsize\titraggedcenter}}
\@ifundefined{abstractheadfont}{}{\def\abstractheadfont{\normalfont\normalsize\bfseries\titraggedcenter}}
\@ifundefined{abstractsubheadfont}{}{\def\abstractsubheadfont{\normalfont\normalsize\bfseries}}
\@ifundefined{abstractfont}{}{\def\abstractfont{\normalfont\normalsize\leftskip=24pt\rightskip=24pt\parfillskip=0pt plus 1fil}}
\@ifundefined{keywordfont}{}{\def\keywordfont{\normalfont\normalsize\leftskip=24pt\rightskip=24pt plus .5fil}}
\@ifundefined{subsectionfont}{}{\def\subsectionfont{\normalfont\normalsize\bfseries}}
\@ifundefined{subsubsectionfont}{}{\def\subsubsectionfont{\normalfont\normalsize\bfseries}}
\@ifundefined{figurecaptionfont}{}{\def\figurecaptionfont{\normalfont\normalsize}}
\@ifundefined{tablecaptionfont}{}{\def\tablecaptionfont{\normalfont\normalsize}}

\def\ps@discoverreview{%
  \let\@oddhead\@empty\let\@evenhead\@empty
  \def\@oddfoot{\hfil\normalfont\normalsize\thepage\hfil}%
  \let\@evenfoot\@oddfoot}
\AtBeginDocument{%
  \@ifclassloaded{sn-jnl}{%
    \let\ps@headings\ps@discoverreview
    \let\ps@titlepage\ps@discoverreview
    \let\ps@plain\ps@discoverreview}{}%
  \pagestyle{discoverreview}\raggedbottom\normalsize
  \typeout{DISCOVER_BODY_SIZE_PT=\f@size}}
\makeatother
\hypersetup{colorlinks=true,allcolors=blue,
pdftitle={Quantum walks for tie breaking in spanning tree construction},
pdfauthor={F. S. Luiz et al.},
pdfsubject={Author draft for Discover Quantum Science}}
\setcitestyle{numbers,square,comma,sort&compress}
\newcommand{\ket}[1]{\lvert #1\rangle}

\newcommand{\tr}{\operatorname{Tr}}

\begin{document}
\title[Quantum walks for spanning-tree tie breaking]{Quantum walks for tie breaking in spanning tree construction: quantum implementation and a classical fourth-order reduction}
\author*[1]{\fnm{F. S.}\sur{Luiz}}\email{fsluiz@unicamp.br}
\author[2,3]{\fnm{F. F.}\sur{Fanchini}}
\author[4]{\fnm{Victor Hugo C.}\sur{de Albuquerque}}
\author[2]{\fnm{J. P.}\sur{Papa}}
\author*[1,3]{\fnm{M. C.}\sur{de Oliveira}}\email{marcos@ifi.unicamp.br}
\affil[1]{\orgdiv{Instituto de F\'isica Gleb Wataghin}, \orgname{Universidade Estadual de Campinas}, \orgaddress{\city{Campinas}, \state{SP}, \postcode{13083-859}, \country{Brazil}}}
\affil[2]{\orgdiv{Faculty of Sciences}, \orgname{S\~ao Paulo State University}, \orgaddress{\city{Bauru}, \state{SP}, \postcode{17033-360}, \country{Brazil}}}
\affil[3]{\orgname{QuaTI -- Quantum Technology \& Information}, \orgaddress{\city{S\~ao Carlos}, \state{SP}, \postcode{13560-161}, \country{Brazil}}}
\affil[4]{\orgdiv{Department of Teleinformatics Engineering}, \orgname{Federal University of Cear\'a}, \orgaddress{\city{Fortaleza}, \state{CE}, \postcode{60455-970}, \country{Brazil}}}
\abstract{
Equal-cost choices in spanning-tree construction can affect feasibility and cost when degree bounds restrict later connections. We formulate a hybrid protocol in which continuous-time quantum-walk probabilities rank edges within cost ties, while a classical controller checks cycles and degrees. Keeping original costs as the primary priority preserves the minimum-spanning-tree guarantee when the degree bound is removed. A two-qubit superconducting implementation recovered a prescribed phase-induced change between two degree-feasible trees of equal cost. A separate temporal acquisition gave finite-window quartic-coefficient estimates compatible with analytical predictions, although quartic dominance of one contrast and a hardware equivalence control remained unresolved. For real inverse-cost Hamiltonians, the probability expansion has an inverse-square-cost leading term and a fourth-order correction containing interference between two-step walks and between one-step and three-step amplitudes. Real third-order product formulas preserve this probability coefficient. Evaluating it by matrix products yields the classical rule QW4, with cubic dense preprocessing and conditional short-time ordering equivalence for nondegenerate ties. On 640 small degree-constrained validation cases, QW4 reproduced the trees of the ideal finite-time walk reference; classical comparisons did not establish general cost superiority. Controlled noise simulations show reduced precision for certifying quartic dominance at the adopted sampling budget. The protocol connects a physical phase response and a short-time interference coefficient to an explicit classical heuristic, with separate limits on hardware certification and optimization quality.

}
\maketitle
\thispagestyle{discoverreview}

\section{Introduction}
\label{sec:introduction}

Kruskal's algorithm finds a minimum spanning tree (MST) by processing edges in increasing cost and rejecting cycles~\cite{Kruskal1956}. Equal-cost edges can be processed in different orders without changing the optimal unrestricted cost. With a degree bound, however, a choice can exhaust a vertex's capacity and alter later connections, feasibility and final cost. Bounded-degree spanning-tree optimization is NP-hard; degree two includes the minimum-cost Hamiltonian path problem~\cite{Singh2015}. For a degree-feasible greedy construction, a secondary rule specifies the order within cost ties without changing the unrestricted MST guarantee.

Continuous-time quantum walks assign graph-dependent transition probabilities by combining amplitudes from different walks~\cite{Farhi1998}. Short-time path expansions already relate transition asymptotics to graph paths and coupling phases~\cite{Szigeti2019}. Within an equal-cost class, walks through the rest of the graph supply structural information for a secondary heuristic; its effect on constrained cost must be evaluated separately. Coupling phases can change transport while preserving magnitudes, with cycle phases invariant under local vertex rephasing~\cite{Zimboras2013}. This provides a way to test interference without redefining input costs. Our controls use real signs, corresponding to phases zero and $\pi$, and retain ideal time-reversal symmetry and probabilities even in time.

Quantum MST algorithms have also been studied in adjacency-matrix and adjacency-array query models~\cite{Durr2006}. Their oracle complexities do not account for all costs of a circuit-based heuristic, which also requires Hamiltonian construction, synthesis, preparation, measurement and classical processing. Here we specify these ingredients for a transition-based tie breaker and assess its optimization quality separately.

A related route from quantum dynamics to a classical graph rule was developed for minimum vertex cover (MVC) by Luiz and coauthors~\cite{LuizMVC}: a quantum-walk protocol and hardware implementation were followed by a short-time analytical reduction to a classical vertex criterion. The present MST protocol instead uses a fixed inverse-cost Hamiltonian to score edges. Original costs remain the primary priority; symmetrized transition probabilities resolve ties, and a classical controller checks cycles and degrees. The register encodes walker position, while the forest remains classical. Saturated vertices do not trigger generator updates.

We implement this protocol on a superconducting processor and recover a prescribed phase-induced change between two degree-feasible trees of equal cost. A separate temporal acquisition probes quartic coefficients on an auxiliary square using exact evolution factorizations. Count-likelihood inference with finite shared calibration gives finite-window coefficient intervals compatible with the analytical values, while one dominance criterion and a hardware equivalence control remain unresolved.

The real-Hamiltonian expansion identifies an inverse-square-cost quadratic term and a fourth-order correction containing two-step and one--three interference. We determine when real product formulas preserve this probability coefficient~\cite{Hatano2005,Childs2021}, then evaluate it directly to obtain QW4, a classical rule with cubic dense preprocessing. A fixed input with nondegenerate coefficient gaps admits a short-time ordering equivalence. On 640 small degree-bounded cases, QW4 and the ideal finite-time reference return the same trees, although some edge priorities differ. Classical comparisons do not establish general cost superiority. The quantum-inspired designation describes the score's derivation; its heuristic quality and hardware certification require their own evidence.

Section~\ref{sec:protocol} defines the protocol, Section~\ref{sec:implementation} presents its implementation, Section~\ref{sec:shorttime} derives and tests the short-time terms, and Section~\ref{sec:classical} develops the classical reduction and benchmarks. Section~\ref{sec:discussion} evaluates the scope of the results. Proofs, reproduction methods and extended diagnostics are provided in the Supplementary Material.

\section{A quantum-walk protocol for spanning-tree construction}
\label{sec:protocol}

Let $G=(V,E)$ be a connected undirected simple graph with $N\geq2$ vertices and positive costs $w_{ij}=w_{ji}$. We seek a spanning tree minimizing
\begin{equation}
 \mathcal C(T)=\sum_{\{i,j\}\in T}w_{ij},
 \qquad d_T(v)\leq\Delta.
 \label{p:objective}
\end{equation}
Here $d_T(v)$ is the degree of vertex $v$ in the tree and $\Delta$ is its allowed maximum. Original cost is the primary priority; transition probabilities resolve equal-cost choices. Equality is determined by the input representation, without merging nearby unequal costs.

\subsection{Fixed generator and measured score}

Define the inverse-cost adjacency and its common input scale,
\begin{align}
 B_{ij}&=w_{ij}^{-1}\quad(\{i,j\}\in E),\qquad
 B_{ii}=0,\qquad \Lambda_G=\max_i\sum_jB_{ij},\label{p:adjacency}\\
 H_0&=-B/\Lambda_G,\qquad U_0(\tau)=e^{-i\tau H_0}.
 \label{p:H}
\end{align}
Nonedges have $B_{ij}=0$. The generator satisfies $\|H_0\|_2\leq1$ and remains fixed throughout construction. For a physical Hamiltonian $\mathcal H=\hbar\gamma H_0$, $\tau=\gamma t$. Dividing by $\Lambda_G$ rescales energies and time; an additive scalar multiple of the identity instead changes only global phase.

For each known source $j$, prepare $\ket j$, evolve and measure in the vertex basis. Define
\begin{equation}
 P_{ij}(\tau)=|[U_0(\tau)]_{ij}|^2,\qquad
 s_{ij}(\tau)=\tfrac12[P_{ij}(\tau)+P_{ji}(\tau)].
 \label{p:probability}
\end{equation}
The ideal real symmetric generator gives reciprocal probabilities. We write $\widetilde P_{ij}$ for the experimental estimate of the transition probability from source $j$ to outcome $i$, and use $\widetilde s_{ij}=\tfrac12(\widetilde P_{ij}+\widetilde P_{ji})$. Without calibration correction, $\widetilde P_{ij}=n_{i|j}/S_j$, where $n_{i|j}$ is the corresponding count among $S_j$ shots. For the phase-control experiment below, $\widetilde P=\widehat A^{-1}\widehat F$: $\widehat F$ contains the frequencies after undoing the readout masks, and $\widehat A$ is the mean of the adjacent calibration matrices (Supplementary Section~\mbox{S2.4}, Eq.~\textup{(S12)}). Calibration uncertainty is propagated jointly to the contrast intervals.

All outcomes, including returns, enter each source's shot denominator. Source labels are essential: a uniformly mixed input with its source discarded remains $I/N$.

Binary encoding uses $n=\lceil\log_2N\rceil$ qubits and zero-pads the generator when necessary. Unused outcomes remain in the shot denominator. The register encodes walker position; the forest, components and degree counters reside in classical memory.

\subsection{Controller and phase controls}

Order edges lexicographically by
\begin{equation}
 (w_{ij},-\widetilde s_{ij},i,j),\qquad i<j.
 \label{p:ordering}
\end{equation}
Starting with an empty forest, accept an edge only if it joins different components and both endpoint degrees are below $\Delta$. Stop at $N-1$ accepted edges; otherwise report an incomplete forest. Saturation rejects later incident edges without updating the Hamiltonian or probability table.

Every completed output is a degree-feasible spanning tree. Omitting the degree test gives Kruskal's algorithm with a secondary order inside cost ties, and therefore an MST for any finite score matrix~\cite{Kruskal1956}. With the degree bound active, the rule is a heuristic; failure to complete does not certify infeasibility. Supplementary Fig.~\mbox{S1} separates the quantum score acquisition from the classical controller; pseudocode and proofs appear in Supplementary Section~\mbox{S1}.

To test the transition mechanism, change real coupling signs while preserving magnitudes, original costs and $\Lambda_G$. Write $H_\phi$ for the baseline ($\phi=0$) or sign-flipped ($\phi=\pi$) generator. Flipping the sign of an edge on a cycle can change transition probabilities. A vertex gauge $H_g=QH_0Q$, with diagonal $Q_{jj}=\pm1$, preserves them because $U_g=QU_0Q$. For tied edges $e,f$, define
\begin{equation}
 C_\phi=s_e^{(\phi)}-s_f^{(\phi)},\qquad W=C_0-C_\pi.
 \label{p:W}
\end{equation}
The signs of $C_0,C_\pi$ identify the preferred edge under the baseline and phase-$\pi$ control. The gauge contrast ideally satisfies $C_g=C_0$. These observables test the dynamics separately from output-tree cost.

\section{Quantum implementation and physical acquisitions}
\label{sec:implementation}

Two acquisitions address different questions: a weighted complete graph tests phase-controlled tree choice, and an auxiliary square tests the short-time probability law. Supplementary Table~\mbox{S1} summarizes their inputs, circuits and endpoints.

\subsection{Circuit realization}

For $n$ qubits, let $P_\ell\in\{I,X,Y,Z\}^{\otimes n}$ be a Pauli word, with $I$ the identity and $X,Y,Z$ the Pauli matrices. The rightmost factor acts on logical qubit 0. The padded Hamiltonian has the exact expansion
\begin{equation}
 H_\phi=\sum_\ell h_{\ell,\phi}P_\ell,\qquad
 h_{\ell,\phi}=2^{-n}\tr(P_\ell H_\phi).
 \label{q:pauli}
\end{equation}
Products of Pauli rotations approximate its exponential; the basis expansion itself introduces no temporal approximation. The selected two-qubit phase test instead used direct synthesis of the classically computed matrix $e^{-i\tau H_\phi}$. All seven compiled evolutions used three controlled-$Z$ (CZ) gates interleaved with local rotations (Fig.~\ref{q:circuit}); their maximum entrywise discrepancy from the ideal targets was below $3.3\times10^{-14}$. This checks circuit specification, not hardware fidelity, and includes classical matrix construction in the preparation cost.

\begin{figure}[htbp]
\centering
\includegraphics[width=.92\linewidth]{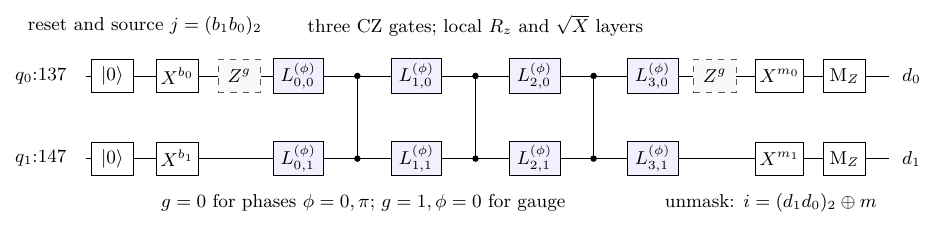}
\caption{Two-qubit implementation with four local layers and three CZ connections. Dashed gauge gates reuse the baseline evolution. Source and readout-mask labels are retained. Mapping, Pauli coefficients and resources appear in Supplementary Section~\mbox{S2}; Section~\mbox{S2.7} identifies the archived numerical gate parameters}
\label{q:circuit}
\end{figure}

\subsection{Phase-controlled tree choice}

The selected complete four-vertex input has
\begin{equation}
 (w_{01},w_{02},w_{03},w_{12},w_{13},w_{23})
 =\tfrac17(140,45,83,102,26,83),\qquad \Delta=2,
 \label{q:weights}
\end{equation}
with $\Lambda_G=0.4379054680$. The phase control flips coupling 03, and $C_\phi=s_{03}^{(\phi)}-s_{23}^{(\phi)}$. This input was chosen before acquisition from 2,048 calibration examples using ideal visibility, a predicted tree change and an offline noise-model screen. It is a selected mechanism demonstration.

IBM Marrakesh executed the circuits on qubits 137/147 on 30 October 2026. The primary time was $\tau_*=1.3454342644$, with secondary times $3\tau_*/4$ and $\tau_*/2$. Of the nine jobs submitted once, seven returned a total of 491,520 shots. The two failed jobs left the short-time block unavailable; reported usage was 170 s including these failures.

For known-source measurements, we used four readout masks, undoing the bit flips in the outcome labels, and calibration packs placed before and after the measurement blocks. Assignment inversion propagated shared calibration uncertainty through 10,000 stratified bootstrap replicates. The three primary contrasts used nominal individual 98.333\% confidence intervals for the prespecified 95\% family. An earlier pilot did not validate its four primary contrasts; its separate records and the full methods are given in Supplementary Section~\mbox{S2}.

At $\tau_*$, all three contrast signs were resolved as predicted, selecting
\begin{equation}
 T_0=\{02,13,23\},\qquad T_\pi=\{02,03,13\},
 \qquad \mathcal C(T_0)=\mathcal C(T_\pi)=22.
 \label{q:trees}
\end{equation}
Only $W=-0.029870$, with primary interval $[-0.034903,-0.025037]$, covered its complete ideal value $-0.029162$; both individual contrasts were biased (Fig.~\ref{q:results}). The gauge interval lay within its prescribed margin $\pm0.02$ but excluded zero; this absolute margin does not certify relative accuracy of the small $C_0$ contrast. At the intermediate time, $C_\pi$ had the opposite sign to ideal and $W$ was unresolved. These limits preclude extending the primary observation uniformly across time or treating equal tree costs as an optimization improvement.

\begin{figure}[!htb]
\centering
\includegraphics[width=\linewidth]{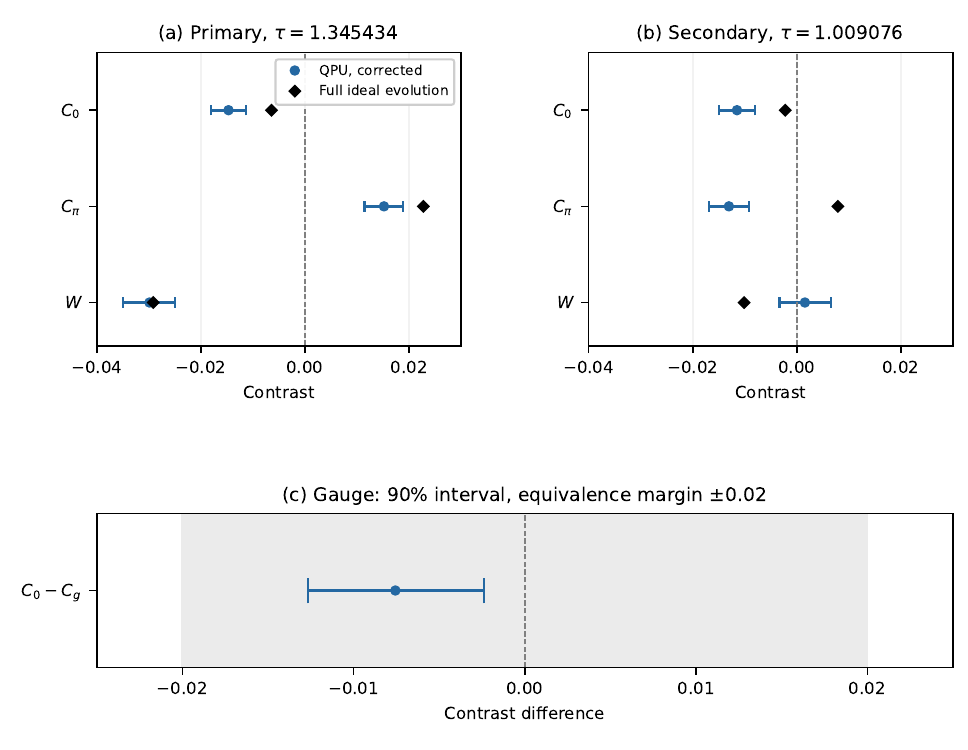}
\caption{Phase-control observations. Primary intervals are nominal individual 98.333\%; intermediate intervals are 95\%. Diamonds show complete ideal evolution. The gauge interval passes its absolute margin while excluding zero. Both complete physical blocks are shown; the short-time block returned no counts}
\label{q:results}
\end{figure}

\subsection{A separate temporal acquisition}
\label{q:temporalacquisition}

The finite-time tree-choice result does not isolate a quartic coefficient. A second acquisition therefore used an auxiliary unit-cost square with edges $01,02,13,23$ and $\Lambda=2$:
\begin{equation}
 H_0^\square=-\frac{IX+XI}{2},\qquad
 H_\pi^\square=-\frac{IX+XZ}{2}.
 \label{q:temporalH}
\end{equation}
Both evolutions admit exact factorizations. A shared template, compiled before binding time and phase and retaining zero-angle gates, used four CZ and ten $\sqrt X$ gates. Thus its ideal reference has no finite-step Trotter error.

The confirmation ran on Marrakesh qubits 54/55 on 2 November 2026, after qualification independent of quartic-coefficient estimates. It sampled sources 0/3, both time signs, two phases and four masks at $|\tau_k|=7k/60$, $k=0,\ldots,12$. Four blocks shared five calibration packs. All nine jobs completed, returning 4,620,288 shots, including 4,194,304 main-series shots, with 1,394~s of reported quantum processing unit (QPU) usage. Pilots were excluded, and no optional stopping or extra shots were used. Submitted circuits, raw bits and histograms were independently reconciled. The statistical revision below uses these archived counts; exact circuit identities and the acquisition record are in Supplementary Section~\mbox{S2}.

\section{Short-time probabilities and fourth-order interference}
\label{sec:shorttime}

\subsection{Expansion and digital preservation}

For a real symmetric generator and $i\ne j$, the amplitude expands as
\begin{equation}
 [e^{-i\tau H}]_{ij}=-i\tau H_{ij}
 -\frac{\tau^2}{2}(H^2)_{ij}
 +\frac{i\tau^3}{6}(H^3)_{ij}+O(\tau^4).
 \label{s:amplitude}
\end{equation}
Reality gives $P_{ij}(-\tau)=P_{ij}(\tau)$, hence
\begin{align}
 P_{ij}(\tau)&=\tau^2H_{ij}^2+\tau^4c_{ij}+O(\tau^6),\label{s:probability}\\
 c_{ij}&=\frac14[(H^2)_{ij}]^2-\frac13H_{ij}(H^3)_{ij}.
 \label{s:c4}
\end{align}
The first contribution to $c_{ij}$ combines two-step amplitudes; the second combines one-step and three-step amplitudes. Both are required to recover the coherent transition probability. Hermiticity alone would not remove odd powers for arbitrary complex couplings.

For $H=-B/\Lambda_G$, the leading term is $\tau^2/(\Lambda_G^2w_{ij}^2)$, recovering the inverse-cost preference. Equal-cost edges share this term, so their probability difference begins with $\tau^4(c_e-c_f)$. A real sign flip preserves the quadratic term and can change the quartic one. Consequently,
\begin{equation}
 W(\tau)=a_4\tau^4+O(\tau^6),\qquad
 a_4=(c_e^{(0)}-c_f^{(0)})-(c_e^{(\pi)}-c_f^{(\pi)}).
 \label{s:W}
\end{equation}
The unrestricted MST guarantee still follows from the primary cost key, independently of finite-time probability ordering.

Digital accuracy can be assessed at this coefficient. For a consistent fixed product of exponentials with real Hermitian factors and real schedule coefficients, write $V(\tau)=\sum_k(-i\tau)^kR_k$, with real $R_k$ and $R_1=H$. Its quartic probability coefficient is
\begin{equation}
 c_{ij}^{V}=(R_2)_{ij}^2-2H_{ij}(R_3)_{ij}.
 \label{s:digitalc4}
\end{equation}
A third-order operator formula has $R_2=H^2/2$ and $R_3=H^3/6$, preserving $c_{ij}$; its first possible probability error is sixth order despite a fourth-order operator error. A generic second-order formula need not preserve it. Product schedules, recurrence checks and a uniform remainder bound are given in Supplementary Section~\mbox{S3}~\cite{Hatano2005,Childs2021}.

The original phase experiment is a finite-time test. Its exact coefficient is approximately $-0.01098640$, but at $\tau_*$ the quartic truncation gives $W=-0.03600023$ against the full value $-0.02916240$, a 23.45\% difference. Only the full value lies in the primary experimental interval. Numerical convergence to the short-time law is documented in Supplementary Section~\mbox{S3}. This difference motivates the separate temporal test below rather than a coefficient estimate from that single time.

\subsection{Auxiliary temporal observables}
\label{s:temporaltest}

For the square in Eq.~\eqref{q:temporalH}, define
\begin{align}
 D(\tau)&=\frac12\sum_{j\in\{0,3\}}
 (P_{3-j,j}^{(0)}-P_{3-j,j}^{(\pi)})
 =\sin^4(\tau/2),\label{s:squareD}\\
 E(\tau)&=\frac12\sum_{j\in\{0,3\}}\sum_{i\in\{1,2\}}
 (P_{ij}^{(0)}-P_{ij}^{(\pi)})
 =\frac12\sin^2\tau-\sin^2(\tau/\sqrt2).
 \label{s:squareE}
\end{align}
Thus $D=\tau^4/16+O(\tau^6)$ and $E=-\tau^4/12+O(\tau^6)$. The opposite-vertex observable D probes two-step interference; the neighbour observable E cancels the quadratic term and retains one--three interference. Neither is a tree-cost comparison.

We compare the IBM observations with a simulated FakeMarrakesh dataset using the same circuits and shot allocation. The Fake channels use the archived 26 February 2025 calibration, rather than the calibration at hardware acquisition. Here ``dataset'' distinguishes Fake and IBM; ``source'' denotes the prepared vertex $j$.

The post-acquisition revision fits raw four-category multinomial histograms jointly with their finite shared calibrations. For each observable, $y_b(\tau)$ denotes the expected assignment-corrected phase contrast, averaged over sources 0/3 and both time signs in temporal block $b=1,\ldots,4$. Its model is
\begin{equation}
 y_b(\tau)=b_{0,b}+b_{2,b}\tau^2+b_4\tau^4+b_6\tau^6+b_8\tau^8,
 \label{s:temporal_model}
\end{equation}
The intercept and quadratic term vary by block, whereas the fourth-, sixth- and eighth-order terms are common across blocks. All coefficients are free. The fitted $b_4$ is a finite-window coefficient: higher temporal orders can contribute over the fitted window, so it need not equal a Taylor coefficient even for ideal unitary evolution. Assignment-corrected contrasts are operational quantities, rather than latent unitary probabilities or binomial counts.

The revised count-likelihood analysis addresses rare-event variance instability in the original empirical treatment without changing the counts, windows or decision margins. We report nominal individual 97.5\% profile intervals based on the likelihood-ratio cutoff and assess goodness of fit with 2,000 parametric bootstrap replicates per dataset/observable~\cite{Wilks1938,Davison1997}. The full likelihood, controls, original results and conditional validation are given in Supplementary Section~\mbox{S4}.

The coefficient criteria require the predicted sign to be resolved, inclusion of the analytical value $a_4$, relative interval half-width $(u-l)/(2|a_4|)\leq0.5$, and $p_{\rm GoF}\geq0.025$. Quartic dominance within the fitted model is a separate endpoint: $R_{.95}$ is the 95th percentile of the bootstrap distribution of the largest conservative nonquartic-to-quartic ratio across the four blocks and $\tau\in\{0.7,0.816667,0.933333\}$. It must be below 0.5 (Supplementary Eq.~\textup{(S29)}).

\begin{figure}[htbp]
\centering
\includegraphics[width=\linewidth]{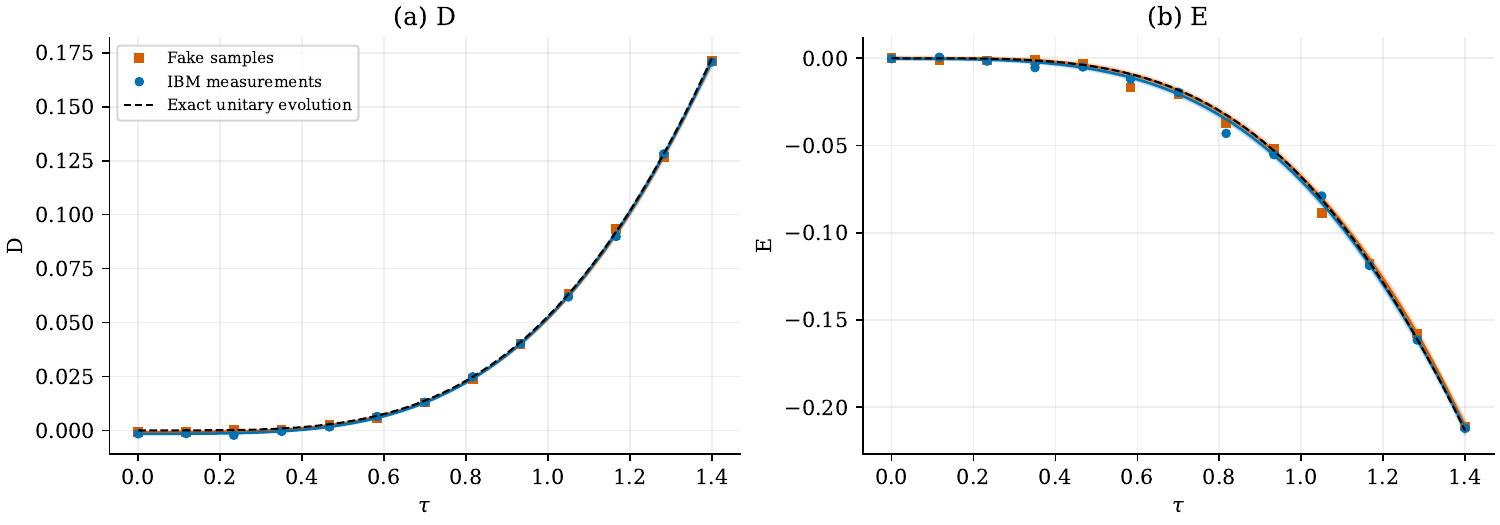}
\caption{Auxiliary temporal experiment. Each point is the mean contrast across the four blocks; the likelihood retains the block-resolved histograms. Solid curves show the average fitted block mean, and shading gives its pointwise 95\% parametric-bootstrap interval. Dashed curves show the complete unitary laws. Fake points are simulated samples; IBM points are archived hardware observations. These pointwise bands provide neither simultaneous coverage nor intervals for individual observations}
\label{s:temporaldata}
\end{figure}

\begin{table}[htbp]
\centering\small

\begin{tabular}{llrrrr}
\toprule
Dataset & Obs. & $\widehat b_4$ & Profile interval (97.5\%) & $p_{\rm GoF}$ & $R_{.95}$ \\
\midrule
Fake & D & 0.05613 & $[0.04424, 0.06801]$ & 0.814 & 0.453 \\
Fake & E & -0.09226 & $[-0.12846, -0.05606]$ & 0.612 & 0.591 \\
IBM & D & 0.05942 & $[0.04815, 0.07069]$ & 0.337 & 0.476 \\
IBM & E & -0.07620 & $[-0.10973, -0.04268]$ & 0.495 & 0.826 \\
\bottomrule
\end{tabular}

\caption{Revised temporal inference. Predictions are $1/16$ for D and $-1/12$ for E. Profile intervals are nominal individual 97.5\%; $p_{\rm GoF}$ uses the parametric bootstrap. Quartic dominance in the fitted model requires $R_{.95}<0.5$.}
\label{s:quarticresults}
\end{table}

Both datasets pass the finite-window coefficient criteria and all four blocks pass qualification (Fig.~\ref{s:temporaldata}, Table~\ref{s:quarticresults}). Coefficient compatibility does not establish a leading quartic law. D passes the dominance endpoint, with Monte Carlo precision intervals reported in Supplementary Table~\mbox{S7}, but E gives $R_{.95}=0.591$ in Fake and 0.826 in hardware.

All six Fake equivalence controls pass; hardware passes five. For the remaining D time-sign control, the interval $[0.01297686,0.02478313]$ extends beyond the unchanged margin $\pm0.02401000$, although its estimate 0.01887963 is inside. This fails to establish equivalence without proving that the true difference exceeds the margin. Full physical leading-law certification remains unsatisfied. Controlled noise and precision diagnostics are reported in Supplementary Section~\mbox{S4.7}.

\section{Classical reduction and independent benchmarks}
\label{sec:classical}

\subsection{QW4 and its short-time domain}

The classical reduction uses the ideal coefficient derived above, independently of whether hardware satisfies every leading-law criterion. The baseline coefficient can be computed without time evolution:
\begin{equation}
 M_2=B^2,\quad M_3=M_2B,\quad
 c_{ij}=\Lambda_G^{-4}\left[\frac{(M_2)_{ij}^2}{4}
 -\frac{B_{ij}(M_3)_{ij}}{3}\right].
 \label{c:matrix}
\end{equation}
The square of an entry is distinct from a matrix product. QW4 sorts edges by $(w_{ij},-c_{ij},i,j)$ and uses the same cycle/degree controller. Its coefficient matrix and order are static and reusable across degree bounds.

For a fixed input whose fourth-order coefficients are distinct within every cost class, QW4 and exact finite-time probabilities have identical priorities at sufficiently short nonzero times. Finitely many nonzero gaps admit a common window by the remainder bound; the deterministic controller then gives the same accepted sequence. The proof and numerical-error margins are in Supplementary Section~\mbox{S3}. Degenerate coefficients, measurement errors and finite-time inversions lie outside this ideal guarantee.

With conventional dense fixed-precision arithmetic, preprocessing costs $O(N^3)$; sorting $M=|E|$ edges costs $O(M\log M)$ and the forest pass costs $O(M\alpha(N))$ with disjoint sets, where $\alpha$ is the inverse Ackermann function~\cite{Tarjan1975}. Storage is $O(N^2+M)$. Access to a cached coefficient is $O(1)$, rather than the entire computation. Exact rational arithmetic adds bit costs. This is an operation-count analysis, not measured runtime scaling. Supplementary Section~\mbox{S5} gives score resolutions, exact-arithmetic reconstruction and reproduction details.

\subsection{Validation and constrained cost}

We first test agreement with the finite-time walk reference, then compare constrained costs across classical tie breakers. The fixed validation contains 80 latent bases from uniform-edge and geometric families at $N=4,6,7,8$. Four continuous/quantized weight versions give 320 complete inputs and 640 cases at $\Delta=2,3$. For $N=4$, $\Delta=3$ is already unrestricted. Exact constrained optima $\mathcal C^*$ were enumerated for evaluation only; gaps are $100(\mathcal C/\mathcal C^*-1)\%$. The physical examples are separate from this ensemble. Paired nominal 95\% intervals use 10,000 latent-base bootstrap resamples, retaining dependence between weight versions and degree bounds; all generation and evolution settings are in Supplementary Section~\mbox{S5}.

The primary full propagator is the complete continuous evolution $e^{-i\tau H_0}$. The primary product uses a finite rotation schedule at the same registered time: one Ruth step (or the exact single-group product) for four vertices, and one five-block fourth-order product for six to eight vertices. Times and schedules are specified in Supplementary Section~\mbox{S5.2}. QW4 and the primary full propagator returned identical trees in all 640 cases. The primary product had the same cost in 640, the same tree in 603 and the same sequence in 600. Nevertheless, two nondegenerate pair priorities differed between QW4 and the continuous reference, without changing a tree. Every first divergent accepted edge between product and continuous constructions had exactly degenerate quartic coefficients. Tree agreement therefore establishes neither universal priority agreement nor a uniform time window.

We compare seven rules, all with the same cost-first controller: QW4, stable index order, inverse-cost strength, two coefficient ablations, one random order and best-of-30 random orders. With $a_i=\sum_jB_{ij}/\Lambda_G$, strength uses score $-(a_i+a_j)$. The ablations retain the two summands of $c_{ij}$ separately. Best-of-30 selects the least costly of 30 constructions and has a larger search budget.

\begin{table}[htbp]
\centering\small

\begin{tabular}{lrrrrrr}
\toprule
Method & $N=4$ & $N=6$ & $N=7$ & $N=8$ & Pooled & Optima/640 \\
\midrule
Stable & 0.971 & 3.934 & 3.988 & 3.593 & 3.121 & 546 \\
Vertex strength & 0.010 & 0.934 & 1.990 & 1.801 & 1.184 & 567 \\
Two-step term & 1.151 & 3.537 & 6.242 & 10.356 & 5.321 & 511 \\
One--three term & 0.263 & 0.934 & 1.965 & 1.714 & 1.219 & 567 \\
QW4 & 0.727 & 1.393 & 2.745 & 3.919 & 2.196 & 542 \\
Random (one) & 0.341 & 1.579 & 2.785 & 5.888 & 2.648 & 546 \\
Random (best of 30) & 0.010 & 0.492 & 1.637 & 1.115 & 0.814 & 580 \\
\bottomrule
\end{tabular}

\caption{Mean optimality gaps (\%) relative to enumerated constrained optima, pooling $\Delta=2,3$ with 160 cases per size. The last column counts attained optima. Best-of-30 uses 30 constructions; other rules select one. Ablations retain only one of the two contributions to $c_{ij}$.}
\label{c:quality}
\end{table}

Table~\ref{c:quality} retains all seven rules. QW4 attains 542 optima and a mean gap of 2.196\%. Its paired difference from stable order is $-0.925$ percentage points, with interval $[-2.195,0.170]$. Against strength it is $+1.012$ points, with interval $[0.445,1.705]$. The cheaper strength score performs better on this ensemble. The one--three ablation also has a smaller mean gap, while best-of-30 uses a larger search budget. Recovering the full physical coefficient does not establish cost superiority.

All seven rules preserve optimal unrestricted cost. On the 160 bounded-degree cases with distinct continuous costs they have identical sequences because the secondary key is unused. Supplementary Section~\mbox{S5.5} also reports the descriptive subset of 420 quantized-cost cases with $\Delta<N-1$, retaining all seven controls. These boundaries, the small graph sizes and the two input families limit generalization.

Finite-shot simulations further separate output cost from mechanism resolution. Here a contrast is $s_e-s_f$ for equal-cost edges with unequal exact quartic coefficients; resolution requires an interval excluding zero in the direction of the continuous reference. At 131,072 shots per source, direct synthesis confidently resolved 34.94\% of four-vertex contrasts; at six vertices, direct and product resolved one and three of 110,400 trials, respectively, and neither resolved any at seven or eight vertices. These fixed-budget ideal simulations also use size-dependent times, so the resolution comparison does not isolate dependence on $N$. They are not device data or an impossibility theorem. Pair eligibility, numerical reference signs, simultaneous intervals and trial denominators are defined in Supplementary Section~\mbox{S5.6}.

\section{Discussion}
\label{sec:discussion}

The protocol connects coherent propagation, physical implementation and an explicit classical score. The phase-control experiment changes a tied decision while preserving original costs; its two output trees have equal cost. Only the differential contrast agrees quantitatively with the complete ideal prediction, while individual biases, a nonzero gauge offset and mixed time dependence remain. The agreement of a tree or its cost therefore provides weaker evidence than agreement of the transition observables.

Building on short-time path expansions~\cite{Szigeti2019}, we use the full quartic coefficient for equal-cost edge priorities. For real third-order product formulas, preservation of this coefficient follows from the amplitude expansion and the time parity of the probabilities. Our contributions are the application of this score to tie breaking, its physical implementation and the assessment of its classical reduction. Both interference terms are retained. This observable-specific result guarantees neither useful native depth nor experimental precision. Direct synthesis here requires classical target construction; input loading, compilation, source-resolved measurements and classical control also contribute to the resource cost. No quantum speedup is established.

The auxiliary acquisition gives finite-window coefficient estimates compatible with theory, without satisfying the full leading-law protocol. Two supplementary diagnostics distinguish count fluctuations from noise-dependent precision (Supplementary Section~\mbox{S4.7}). Controlled expected-histogram comparisons with the same finite budget pass the E dominance criterion in the ideal model and fail it under Fake noise. These histograms remove a particular count fluctuation while retaining finite-budget uncertainty. Known-probability calculations give a population nonquartic fraction near 0.181 in both models: modeled noise reduces certification precision without removing dominance of the mean signal. A sampled ideal realization also failed E in the original analysis. Hardware observations are consistent with a noise limitation, but neither an individual rejection nor the historical Fake snapshot identifies an exclusive device error source.

More shots reduce sampling uncertainty while channel bias and finite-window projection error persist. Temporal sensitivity and conditional coverage therefore matter. Joint coefficient inclusion was 93.4\% in 500 Fake experiments against a 95\% target (Monte Carlo interval 90.86--95.41\%); exact finite-sample family error control is not established. A systematic quartic term also cannot be separated from the desired coefficient by the time fit alone. The diagnostic fits do not replace the main analysis.

QW4 makes the quantum-derived score classically computable, with conditional short-time equivalence and cubic preprocessing. Its small-graph benchmarks favour some simpler controls, so reproducing the physical coefficient is distinct from improving constrained cost. General graphs can also yield incomplete forests under degree bounds. The present evidence concerns a static generator; entropy under a specified channel and degree-dependent extensions require separate analysis (Supplementary Section~\mbox{S6}).

\section{Conclusions}
\label{sec:conclusions}

A fixed quantum walk supplies secondary edge priorities to a classical spanning-tree controller. A two-qubit implementation recovered a prescribed phase-dependent tree choice at equal cost, and a separate temporal experiment gave finite-window quartic-coefficient estimates compatible with theory, with dominance and control limitations retained.

The fourth-order coefficient yields QW4 and an observable-specific product-formula condition. Conditional ordering equivalence and agreement on 640 small validation cases support this reduction, while degeneracies, finite sampling and stronger classical comparators delimit its practical claims. The result is a reproducible classical heuristic derived from coherent dynamics, without general cost superiority or complete physical certification of the leading quartic law.

\section*{Supplementary Information}
Supplementary Information 1 (\texttt{Supplementary\_Information.pdf}) contains the extended protocol, circuit and acquisition records, proofs and higher temporal orders, statistical validation, classical benchmarks and entropy discussion.

\section*{Statements and declarations}

\subsection*{Corresponding authors}
Correspondence to \href{mailto:fsluiz@unicamp.br}{Fabricio de Souza Luiz} and \href{mailto:marcos@ifi.unicamp.br}{Marcos Cesar de Oliveira}.

\subsection*{Funding}
This study was financed, in part, by the São Paulo Research Foundation (FAPESP), Brasil, process no. 2026/04387-7. F. S. L. and M. C. O. acknowledge financial support from the Postdoctoral Researcher Program of the State University of Campinas (UNICAMP). M. C. O. also acknowledges partial financial support from the National Council for Scientific and Technological Development (CNPq). The authors acknowledge support from the Office of Naval Research (ONR Global) and the Air Force Office of Scientific Research (AFOSR) under grant N62909-24-1-2012.

\subsection*{Competing interests}
The authors declare no competing interests.

\subsection*{Data availability}
The graph inputs, reference optima, hardware histograms and bitstrings, calibrations, simulation records and inferential outputs are publicly available in the accompanying versioned repository. Study data are released under CC BY 4.0; third-party notices are retained.~\cite{LuizMSTCode2026}

\subsection*{Code availability}
The implementation, circuit-construction and analysis scripts are available in the same repository. Code written for this study is released under the MIT license. The README and REPRODUCIBILITY.md provide installation and result-recovery commands.~\cite{LuizMSTCode2026}

\subsection*{Ethics approval and consent}
Not applicable. This study concerns weighted graphs and quantum-processor measurements and does not involve human participants or animal research.

\subsection*{Acknowledgements}
The authors thank Professor Eduardo Miranda (Department of Condensed Matter Physics) for insightful discussions on evolution time that greatly benefited this work.

\bibliography{references}

@article{Farhi1998,
  author = {Farhi, Edward and Gutmann, Sam},
  title = {Quantum computation and decision trees},
  journal = {Physical Review A},
  volume = {58}, number = {2}, pages = {915--928}, year = {1998},
  doi = {10.1103/PhysRevA.58.915},
  url = {https://arxiv.org/abs/quant-ph/9706062}
}

@article{Kruskal1956,
  author = {Kruskal, Joseph B.},
  title = {On the Shortest Spanning Subtree of a Graph and the Traveling Salesman Problem},
  journal = {Proceedings of the American Mathematical Society},
  volume = {7}, number = {1}, pages = {48--50}, year = {1956},
  doi = {10.1090/S0002-9939-1956-0078686-7}
}

@incollection{Hatano2005,
  author = {Hatano, Naomichi and Suzuki, Masuo},
  title = {Finding Exponential Product Formulas of Higher Orders},
  booktitle = {Quantum Annealing and Related Optimization Methods},
  editor = {Das, Arnab and Chakrabarti, Bikas K.},
  publisher = {Springer}, year = {2005}, pages = {37--68},
  doi = {10.1007/11526216_2},
  url = {https://arxiv.org/abs/math-ph/0506007},
  note = {\url{https://doi.org/10.1007/11526216_2}}
}

@article{Childs2021,
  author = {Childs, Andrew M. and Su, Yuan and Tran, Minh C. and Wiebe, Nathan and Zhu, Shuchen},
  title = {Theory of {Trotter} Error with Commutator Scaling},
  journal = {Physical Review X}, volume = {11}, pages = {011020}, year = {2021},
  doi = {10.1103/PhysRevX.11.011020},
  url = {https://arxiv.org/abs/1912.08854}
}

@article{Wilks1938,
  author = {Wilks, Samuel S.},
  title = {The Large-Sample Distribution of the Likelihood Ratio for Testing Composite Hypotheses},
  journal = {The Annals of Mathematical Statistics},
  volume = {9}, number = {1}, pages = {60--62}, year = {1938},
  doi = {10.1214/aoms/1177732360}
}

@book{Davison1997,
  author = {Davison, A. C. and Hinkley, D. V.},
  title = {Bootstrap Methods and their Application},
  publisher = {Cambridge University Press},
  year = {1997},
  doi = {10.1017/CBO9780511802843},
  note = {\url{https://doi.org/10.1017/CBO9780511802843}}
}

@article{Singh2015,
  author = {Singh, Mohit and Lau, Lap Chi},
  title = {Approximating Minimum Bounded Degree Spanning Trees to within One of Optimal},
  journal = {Journal of the ACM},
  volume = {62}, number = {1},
  pages = {1:1--1:19},
  year = {2015},
  doi = {10.1145/2629366},
  url = {https://cs.uwaterloo.ca/~lapchi/papers/mbdst.pdf}
}

@article{Durr2006,
  author = {D{\"u}rr, Christoph and Heiligman, Mark and H{\o}yer, Peter and Mhalla, Mehdi},
  title = {Quantum Query Complexity of Some Graph Problems},
  journal = {SIAM Journal on Computing},
  volume = {35}, number = {6}, pages = {1310--1328},
  year = {2006},
  doi = {10.1137/050644719},
  url = {https://arxiv.org/abs/quant-ph/0401091}
}

@article{Zimboras2013,
  author = {Zimbor{\'a}s, Zolt{\'a}n and Faccin, Mauro and K{\'a}d{\'a}r, Zolt{\'a}n and Whitfield, James D. and Lanyon, Ben P. and Biamonte, Jacob},
  title = {Quantum Transport Enhancement by Time-Reversal Symmetry Breaking},
  journal = {Scientific Reports},
  volume = {3}, pages = {2361},
  year = {2013},
  doi = {10.1038/srep02361},
  url = {https://www.nature.com/articles/srep02361}
}

@article{LuizMVC,
  author = {Luiz, F. S. and Iwakami, A. K. F. and Moraes, D. H. and de Oliveira, M. C.},
  title = {Scalable Quantum Walk-Based Heuristics for the Minimum Vertex Cover Problem},
  journal = {arXiv preprint arXiv:2512.02940},
  year = {2025},
  note = {Version 2, 23 May 2026. Accepted for publication in Advanced Quantum Technologies; awaiting production},
  doi = {10.48550/arXiv.2512.02940},
  url = {https://arxiv.org/abs/2512.02940v2}
}

@misc{LuizMSTCode2026,
  author = {Luiz, F. S. and Fanchini, F. F. and de Albuquerque, V. H. C. and Papa, J. P. and de Oliveira, M. C.},
  title = {Quantum-walk spanning-tree code and data},
  year = {2026},
  howpublished = {GitHub repository (2026), commit \texttt{68e01c78}},
  url = {https://github.com/fsluiz/quantum-walk-spanning-tree/tree/68e01c78cbb85762831622df4c126fdd84b49416},
  note = {Code: MIT; study data: CC BY 4.0. Third-party notices retained}
}

@article{Szigeti2019,
  author = {Szigeti, Bal{\'a}zs Endre and Homa, G{\'a}bor and Zimbor{\'a}s, Zolt{\'a}n and Barankai, Norbert},
  title = {Short time behavior of continuous time quantum walks on graphs},
  journal = {Physical Review A},
  volume = {100}, pages = {062320}, year = {2019},
  doi = {10.1103/PhysRevA.100.062320},
  url = {https://arxiv.org/abs/1905.03914v2}
}

@misc{JavadiAbhari2024,
 author = {Javadi-Abhari, Ali and Treinish, Matthew and Krsulich, Kevin and Wood, Christopher J. and Lishman, Jake and Gacon, Julien and Martiel, Simon and Nation, Paul D. and Bishop, Lev S. and Cross, Andrew W. and Johnson, Blake R. and Gambetta, Jay M.},
 title = {Quantum computing with {Qiskit}},
 year = {2024},
 eprint = {2405.08810}, archivePrefix = {arXiv}, primaryClass = {quant-ph},
 doi = {10.48550/arXiv.2405.08810},
 url = {https://arxiv.org/abs/2405.08810}
}

@article{Clopper1934,
 author = {Clopper, C. J. and Pearson, E. S.},
 title = {The Use of Confidence or Fiducial Limits Illustrated in the Case of the Binomial},
 journal = {Biometrika}, volume = {26}, number = {4}, pages = {404--413}, year = {1934},
 doi = {10.1093/biomet/26.4.404}
}

@article{Schuirmann1987,
 author = {Schuirmann, Donald J.},
 title = {A comparison of the {Two One-Sided Tests Procedure} and the {Power Approach} for assessing the equivalence of average bioavailability},
 journal = {Journal of Pharmacokinetics and Biopharmaceutics}, volume = {15}, pages = {657--680}, year = {1987},
 doi = {10.1007/BF01068419}
}

@article{Tarjan1975,
 author = {Tarjan, Robert Endre},
 title = {Efficiency of a Good But Not Linear Set Union Algorithm},
 journal = {Journal of the ACM}, volume = {22}, number = {2}, pages = {215--225}, year = {1975},
 doi = {10.1145/321879.321884}
}

@article{Maciejewski2020,
 author = {Maciejewski, Filip B. and Zimbor{\'a}s, Zolt{\'a}n and Oszmaniec, Micha{\l}},
 title = {Mitigation of readout noise in near-term quantum devices by classical post-processing based on detector tomography},
 journal = {Quantum}, volume = {4}, pages = {257}, year = {2020},
 doi = {10.22331/q-2020-04-24-257}
}

@article{Prufer1918,
 author = {Pr{\"u}fer, Heinz},
 title = {Neuer Beweis eines Satzes {\"u}ber Permutationen},
 journal = {Archiv der Mathematik und Physik}, volume = {27}, series = {3}, pages = {142--144}, year = {1918}
}
\end{document}


\maketitle
\noindent This material contains the extended protocol, circuit and acquisition records,
proofs, temporal inference, classical validation and entropy discussion.
References without an S prefix in prose refer to the main article; linked
equation and section references identify the corresponding document.
\tableofcontents

\section{Protocol specification, encoding and guarantees}
\label{sm:protocol}

\subsection{Preparation and source-resolved measurement}

For each known source $j$, prepare $\ket{j}$, evolve for the same prescribed $\tau$, and measure in the vertex basis. The ideal probability of destination $i$ is
\begin{equation}
 P_{ij}(\tau)=\tr\!\left[\ket{i}\bra{i}
 U_0(\tau)\ket{j}\bra{j}U_0(\tau)^\dagger\right]
 =\left|[U_0(\tau)]_{ij}\right|^2.
 \label{p:probability}
\end{equation}
For every source, $\sum_iP_{ij}=1$. Return outcomes $i=j$ are recorded along with transitions. For $S_j$ repetitions, counts $n_{i|j}$ give the raw estimator
\begin{equation}
 \widehat P_{ij}=n_{i|j}/S_j.
 \label{p:frequency}
\end{equation}
The controller uses the estimate $\widetilde P$ defined in Section~\mbox{2}. Without calibration correction, $\widetilde P=\widehat P$; the phase-control experiment instead uses the corrected matrix in Eq.~\eqref{q:correction}. In either case, $\widetilde s_{ij}=(\widetilde P_{ij}+\widetilde P_{ji})/2$.
All outcomes enter the shot denominator. The source label is retained in every experiment. Indeed, discarding that label and preparing $I/N$ would give $U_0(I/N)U_0^\dagger=I/N$ and uniform destination probabilities, losing the source-dependent information used by the rule.

Binary encoding uses $n=\lceil\log_2N\rceil$ qubits to represent the vertex labels. When $2^n>N$, the ideal generator is $H_0\oplus0_{2^n-N}$, so the extra basis states have zero ideal population for a valid source. Outcomes with unused labels are recorded separately, and frequencies on valid vertices retain the denominator $S_j$. Circuit synthesis and the measurement/calibration procedure specify how the ideal evolution and estimators are realized on a device.

\subsection{Controller flow and pseudocode}

Figure~\ref{p:flow} shows how the source-resolved measurements supply a fixed score table. Algorithm~\ref{p:algorithm} then separates this acquisition from the classical acceptance loop: only the forest, component labels and degrees change during construction.

\begin{figure}[htbp]
\centering
\includegraphics[width=\linewidth]{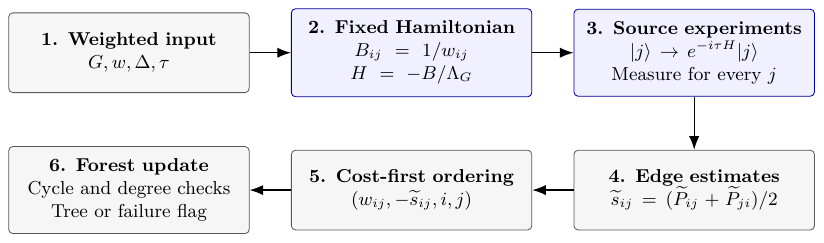}
\caption{The finite-time hybrid protocol. One fixed input Hamiltonian supplies source-resolved transition estimates. The classical controller uses original costs first, scores second, and checks cycles and degrees during forest construction}
\label{p:flow}
\end{figure}

\begin{algorithm}[htbp]
\caption{Finite-time quantum-walk-guided construction}
\label{p:algorithm}
\begin{algorithmic}[1]
\Require $G,w,\Delta,\tau$ and fixed shot counts $S_j$
\State Form $B,\Lambda_G,H_0$ from Eqs.~\textup{(2)}--\textup{(3)}
\For{each source $j\in V$}
  \State Prepare $\ket{j}$, apply $U_0(\tau)$ and measure, for $S_j$ shots
  \State Retain the source label and all measured outcomes
\EndFor
\State Compute the specified estimates $\widetilde P$ and edge scores $\widetilde s$
\State Sort $E$ by Eq.~\textup{(5)}
\State Initialize $F\gets\varnothing$, $d_v\gets0$ and disjoint components
\For{each ordered edge $\{i,j\}$}
  \If{$i,j$ are in different components and $d_i,d_j<\Delta$}
    \State $F\gets F\cup\{\{i,j\}\}$; merge the components; increment $d_i,d_j$
  \EndIf
  \If{$|F|=N-1$}
    \State \Return $F$
  \EndIf
\EndFor
\State \Return the incomplete forest $F$ and a construction-failure flag
\end{algorithmic}
\end{algorithm}

\FloatBarrier
\subsection{Feasibility and the unrestricted limit}

\begin{proposition}
Every completed output of Algorithm~\ref{p:algorithm} is a spanning tree satisfying the stated degree bound. With the degree test omitted, the output is an MST for any finite score matrix.
\end{proposition}
\begin{proof}
Each accepted edge connects two different components, so acyclicity is preserved and the component count decreases by one. Starting from $N$ components, $N-1$ acceptances therefore give a connected spanning tree. The degree counters enforce $d_v\leq\Delta$ at every acceptance. Without the degree test, Eq.~\textup{(5)} processes edges in nondecreasing original cost with an arbitrary order inside ties. This is Kruskal's rule~\cite{Kruskal1956}; the secondary order preserves its MST guarantee.
\end{proof}

With the degree bound active, earlier choices can remove admissible continuations, and the construction is a heuristic. An incomplete forest is reported explicitly; it does not establish that the input has no feasible constrained tree. The unrestricted guarantee follows from the cost priority, so tree cost alone does not test whether the transition scores were correctly implemented. Comparisons of the transition observables provide that additional test.

\subsection{Phase and gauge controls for a tied decision}

For a controlled experiment, keep the original costs and $\Lambda_G$ fixed and assign real signs $\sigma_{ij}=\sigma_{ji}\in\{-1,+1\}$ to the couplings:
\begin{equation}
 H^{(\sigma)}=-\frac1{\Lambda_G}
 \sum_{\substack{\{i,j\}\in E\\i<j}}
 \frac{\sigma_{ij}}{w_{ij}}
 (\ket{i}\bra{j}+\ket{j}\bra{i}).
 \label{p:phaseH}
\end{equation}
The baseline has every $\sigma_{ij}=1$. Changing one sign implements a phase of $\pi$ on that edge while preserving coupling magnitudes, the zero diagonal and the input cost ordering. On a cycle, this flip changes the product of the edge signs and can alter the transition probabilities. Its value is chosen as an experimental control, independently of the degree counters.

For a negative control, choose $Q=\diag(q_0,\ldots,q_{N-1})$, $q_j\in\{-1,+1\}$, and set $H_g=QH_0Q$. Since $Q^2=I$,
\begin{equation}
 U_g(\tau)=QU_0(\tau)Q,
 \qquad |[U_g(\tau)]_{ij}|^2=|[U_0(\tau)]_{ij}|^2.
 \label{p:gauge}
\end{equation}
Thus the ideal probabilities are invariant under this vertex-basis gauge transformation. Products of signs around cycles are also invariant, because each vertex factor appears twice.

For a prespecified equal-cost pair $e,f$, let $C_\phi(\tau)=s_e^{(\phi)}(\tau)-s_f^{(\phi)}(\tau)$ for the baseline or single-edge phase control $\phi\in\{0,\pi\}$. Define the differential observable
\begin{equation}
 W(\tau)=C_0(\tau)-C_\pi(\tau).
 \label{p:W}
\end{equation}
The sign of each $C_\phi$ identifies the preferred member of that tied pair before the feasibility tests. The gauge observable $C_g$ has the ideal equality $C_g=C_0$. These definitions permit separate tests of the phase response, the gauge reference and the resulting forest decision at a fixed evolution time. Their short-time coefficients are derived in Section~\ref{sm:derivation}; the separate temporal experiment and its inference are specified in Sections~\ref{sm:temporalacquisition} and~\ref{sm:temporal}.

\section{Circuit synthesis, acquisition and phase-control records}
\label{sm:implementation}

Table~\ref{q:demonstrations} distinguishes the two reported acquisitions. The weighted graph addresses a phase-controlled decision at finite time; the auxiliary square probes temporal coefficients using a different exact circuit family.

\begin{table}[htbp]
\centering
\begin{tabular}{@{}>{\raggedright\arraybackslash}p{.22\linewidth}>{\raggedright\arraybackslash}p{.34\linewidth}>{\raggedright\arraybackslash}p{.34\linewidth}@{}}
\toprule
Input & Evolution and physical pair & Measured purpose\\
\midrule
Weighted $K_4$ & Direct synthesis, three CZ; qubits 137/147 & Contrasts $C_0,C_\pi,W$ and a tied tree choice at $\tau_*$\\[2pt]
Unit-cost square & Exact factorization, four CZ; qubits 54/55 & Temporal D/E coefficients, dominance and equivalence controls\\
\bottomrule
\end{tabular}
\caption{Roles of the phase-control and temporal experiments. Their datasets, calibrations and inferential endpoints are kept separate.}
\label{q:demonstrations}
\end{table}

\subsection{From the Hamiltonian to a circuit}

On $n=\lceil\log_2N\rceil$ qubits, expand the padded Hermitian generator in Pauli words,
\begin{equation}
 H_\phi=\sum_\ell h_{\ell,\phi}P_\ell,
 \qquad h_{\ell,\phi}=2^{-n}\tr(P_\ell H_\phi).
 \label{q:pauli}
\end{equation}
This is an exact change of operator basis. The temporal approximation enters when the exponentials of noncommuting terms are replaced by a product. For example, writing $H_\phi=A_\phi+B_\phi$, the symmetric second-order formula is
\begin{equation}
 S_2(\delta)=e^{-i\delta A_\phi/2}e^{-i\delta B_\phi}e^{-i\delta A_\phi/2},
 \qquad V_r(\tau)=S_2(\tau/r)^r.
 \label{q:strang}
\end{equation}
More generally, a fixed schedule gives
\begin{equation}
 V_r(\tau)=\left[\prod_{k=1}^L e^{-ia_k(\tau/r)K_k}\right]^r,
 \label{q:product}
\end{equation}
with products acting from right to left. Commuting Pauli terms in a group can be exponentiated separately without group error, using rotations $R_P(\theta)=e^{-i\theta P/2}$. A factor for $hP$ with coefficient $a_k$ therefore uses angle $2a_kh\tau/r$. Higher-order schedules, including Ruth's third-order formula and fourth-order Suzuki compositions, provide alternative digital approximations~\cite{Hatano2005,Childs2021}. Their operator orders alone do not specify the error in a particular measured observable; the probability coefficients relevant to the tie are examined in the analytical section.

For four vertices, the real generator with zero diagonal uses six Pauli words: $IX$, $XI$, $XX$, $XZ$, $YY$ and $ZX$. The rightmost symbol acts on logical qubit 0. The words $IX$, $XI$ and $XX$ commute with one another; $XZ$, $YY$ and $ZX$ form the second commuting group. Table~\ref{q:paulis} gives the coefficients of the physical test instance. This representation supplies a reproducible gate construction and a way to vary the digital approximation without changing the input objective.

\begin{table}[htbp]
\centering

\begin{tabular}{lrr}
\toprule
Pauli word & $h_{\ell,0}$ & $h_{\ell,\pi}$ \\
\midrule
$IX$ & -0.15338624 & -0.15338624 \\
$XI$ & -0.48502058 & -0.48502058 \\
$XX$ & -0.17465505 & +0.01793755 \\
$XZ$ & +0.12979424 & +0.12979424 \\
$YY$ & +0.01793755 & -0.17465505 \\
$ZX$ & +0.03920635 & +0.03920635 \\
\bottomrule
\end{tabular}

\caption{Pauli coefficients for the selected four-vertex input and its single-edge phase control. The normalization and original costs are the same in both columns.}
\label{q:paulis}
\end{table}

For the two-qubit physical demonstration, we used direct synthesis: the target $4\times4$ matrix $e^{-i\tau H_\phi}$ was computed classically and compiled into native gates. The processor then executed this circuit and supplied measurement samples. This small-instance realization includes classical matrix construction and compilation in its preparation costs. Product formulas were studied as implementation controls; the phase-control campaign reported below used direct synthesis.

\subsection{Native circuit and physical mapping}

The phase-control campaign was executed on IBM Marrakesh on 30 October 2026, using physical qubits 137 and 147. Logical qubit 0 maps to physical qubit 137 and classical bit 0; logical qubit 1 maps to physical qubit 147 and classical bit 1. Thus $j=b_0+2b_1$ denotes the prepared state $\ket{b_1b_0}$. Initial and final logical maps are the identity on this pair.

The compiled evolution has the form
\begin{equation}
 U_{\phi,\mathrm{circ}}=e^{i\chi_\phi}
 \mathcal L_3^{(\phi)}\mathrm{CZ}\mathcal L_2^{(\phi)}\mathrm{CZ}
 \mathcal L_1^{(\phi)}\mathrm{CZ}\mathcal L_0^{(\phi)},
 \quad \mathcal L_\ell^{(\phi)}=L_{\ell,1}^{(\phi)}\otimes L_{\ell,0}^{(\phi)},
 \label{q:native}
\end{equation}
where the one-qubit layers use $R_z$ and $\sqrt X$ gates. Figure~\mbox{1} includes preparation, masks and measurement around these layers. Both logical qubits are reset before each shot. The gate angles depend on $\phi$ and $\tau$; the three-CZ pattern is common to the seven evolutions in this package. Table~\ref{q:resources} reports the resulting native resources.

\begin{table}[htbp]
\centering\small

\begin{tabular}{lrrrrr}
\toprule
Evolutions & Number & CZ & One-qubit gates & Depth & Duration (ns) \\
\midrule
Phases $0,\pi$ & 6 & 3 & 29 & 19 & 492 \\
Gauge at $\tau_*$ & 1 & 3 & 31 & 21 & 492 \\
\bottomrule
\end{tabular}

\caption{Resources per compiled evolution. Counts and depth exclude source preparation, reset, masks and measurement; one-qubit counts include virtual $R_z$ gates. Durations are estimates from as-soon-as-possible (ASAP) scheduling using archived backend gate durations, not measured pulse durations.}
\label{q:resources}
\end{table}

For the gauge, $Q=I\otimes Z=\diag(1,-1,1,-1)$. The baseline circuit is reused with $R_z(\pi)$ before and after the evolution on qubit 0, and the global phase is restored in the matrix audit. These additional rotations have zero duration in the archived backend instruction model. All seven ideal compiled operators were compared with their full matrix-exponential targets, including the restored global phase; the maximum entrywise difference was below $3.3\times10^{-14}$. This is a numerical circuit check, distinct from a physical fidelity measurement.

The dimensionless simulated time $\tau$ controls angles. It is distinct from the scheduled duration of the gate circuit: the estimated evolution duration was 492 ns for all seven circuits despite their different $\tau$. The repetition delay was 250 $\mu$s. Compilation used Qiskit 2.6.1~\cite{JavadiAbhari2024} and IBM Runtime 0.51.0; the environment and backend snapshots are preserved in the reproduction records. Ideal output matrices were audited after compilation; a compiler level or gate count was not used as a substitute for that check.

\subsection{Prespecified input, times and acquisition}

The selected $K_4$ instance has costs, in edge order $(01,02,03,12,13,23)$,
\begin{equation}
 (w_{01},w_{02},w_{03},w_{12},w_{13},w_{23})
 =\tfrac17(140,45,83,102,26,83),\qquad \Delta=2.
 \label{q:weights}
\end{equation}
Its scale is $\Lambda_G=0.4379054680$. The tied edges are 03 and 23. The phase-$\pi$ control changes the sign of coupling 03 and preserves all magnitudes, original costs and the scale. We measure $C_\phi=s_{03}^{(\phi)}-s_{23}^{(\phi)}$ and $W=C_0-C_\pi$, as defined in Section~\mbox{2}.

The example was selected before physical acquisition from a new pool of 2,048 calibration inputs, using ideal visibility, a change of the predicted degree-feasible tree, and an offline noise-model screen. Only one of the three selected adjacency/Laplacian cases satisfied the full visibility and bias thresholds; that adjacency case defined the minimal physical package. This selection supports a mechanism demonstration on a chosen input, rather than a general performance claim. The example, phases, times, shots and statistical endpoints were fixed before its hardware results were obtained.

The three prescribed times were $\tau_*/2$, $3\tau_*/4$ and $\tau_*$, where $\tau_*=1.3454342644$ came from the earlier local time-grid study. Every known source was sampled with four readout masks $m\in\{0,1,2,3\}$. For each mask, $X^{m_q}$ acts on logical wire $q$ before measurement, and the measured integer $d$ is interpreted as $i=d\oplus m$. In the $\ket{b_1b_0}$ matrix convention this operation is $X^{m_1}\otimes X^{m_0}$. Mask strata retain their prescribed shot counts.

The primary time used 32,768 shots per source and phase, split equally among masks; the intermediate and short times prescribed 16,384 and 8,192. The gauge at $\tau_*$ used 8,192 per source. Each calibration pack used 4,096 per source, with 1,024 per mask. The acquisition order placed calibration packs before and after each physical block; the post-short pack was also the pre-intermediate pack.

Nine jobs were submitted sequentially, with per-job uniform shots and execution-time caps summing to 178 s. Seven returned 491,520 shots from 144 measured circuits. The pre-short calibration and short-time physical job exceeded their 7 s and 20 s execution caps and returned no counts. They were not repeated. The remaining jobs in the original package were completed with unchanged circuits, order, shots and caps. The reported charge was 170 s including failures; quota and server-response notes are retained in Section~\ref{q:records}.

The primary and intermediate blocks therefore have their required data and calibration packs. The short block is absent. The analysis included complete blocks by data availability, retaining the three-endpoint primary family and the originally planned stability thresholds.

\subsection{Calibration and uncertainty in the phase-control campaign}
\label{sm:calibration}

For a calibration pack $b$, prepare known sources with identity evolution and use the same masks. If $N^{(b)}_{d|j,m}$ is its recorded count, the effective basis-assignment matrix is
\begin{equation}
 \widehat A_{b,ij}=\frac1{S_{\mathrm{cal}}}\sum_m N^{(b)}_{i\oplus m|j,m},
 \qquad S_{\mathrm{cal}}=4096.
 \label{q:calibration}
\end{equation}
For physical block $k$, similarly form the unmasked frequency matrix $\widehat F_{\phi,k}$, with its own total shots per source. The correction uses the average of its before/after packs,
\begin{equation}
 \widehat A_k=\tfrac12(\widehat A_{b^-}+\widehat A_{b^+}),
 \qquad \widetilde P_{\phi,k}=\widehat A_k^{-1}\widehat F_{\phi,k}.
 \label{q:correction}
\end{equation}
Assignment-matrix inversion is a standard readout-mitigation approach~\cite{Maciejewski2020}. This implementation assumes that the calibration transfers to the corresponding physical data. Preparation and assignment errors are not separately identified by these packs. Negative corrected entries, when present, are retained as estimator values.

The 10,000 bootstrap replicates resample counts within each source/mask stratum. Each replicate shares its resampled calibration inverse between phases and the gauge, preserving the resulting covariance in $W$ and $C_0-C_g$. The primary endpoints are $C_0,C_\pi,W$ at $\tau_*$. Their nominal individual 98.333\% intervals implement the prespecified Bonferroni 95\% family procedure. Secondary-time intervals are nominal 95\%. Gauge equivalence uses a 90\% interval for a 95\% TOST~\cite{Schuirmann1987} with absolute margin $\pm0.02$.

Calibration stability is checked by a mask-stratified conditional randomization test and a practical entrywise difference limit of 0.015. The significance threshold remains $0.01/12$, corresponding to the three planned blocks and four sources, even though only two complete blocks were available. The observed maximum paired entry difference was 0.006836 and the largest matrix condition number was 1.044. No available pair triggered the prespecified rejection. This does not certify stationarity or isolate detector error.

\subsection{Physical observations and their scope}

Table~\ref{q:contrasts} and Fig.~\mbox{2} present the corrected results and full ideal evolution. At the primary time, all three signs were resolved in their predicted directions. The observed $C_0<0$ and $C_\pi>0$ selected the two expected degree-two trees,
\begin{equation}
 T_0=\{02,13,23\},\qquad T_\pi=\{02,03,13\},
 \qquad \mathcal C(T_0)=\mathcal C(T_\pi)=22.
 \label{q:trees}
\end{equation}
Their equal costs identify a change of the tied decision, without demonstrating a cost advantage. Only $W$ was quantitatively compatible with its full ideal value in the corresponding primary interval. Both individual contrasts were biased relative to their ideals.

\begin{table}[htbp]
\centering\small

\begin{tabular}{lrrl}
\toprule
Contrast & Full ideal & QPU corrected & Nominal interval \\
\midrule
\multicolumn{4}{l}{Primary: $\tau=1.345434$, individual 98.333\% intervals} \\
$C_0$ & -0.006415 & -0.014688 & [-0.018135, -0.011290] \\
$C_\pi$ & +0.022748 & +0.015182 & [+0.011443, +0.018923] \\
$W$ & -0.029162 & -0.029870 & [-0.034903, -0.025037] \\
\midrule
\multicolumn{4}{l}{Secondary: $\tau=1.009076$, individual 95\% intervals} \\
$C_0$ & -0.002237 & -0.011499 & [-0.014915, -0.008091] \\
$C_\pi$ & +0.007893 & -0.013053 & [-0.016801, -0.009251] \\
$W$ & -0.010130 & +0.001553 & [-0.003335, +0.006528] \\
\bottomrule
\end{tabular}

\caption{Corrected hardware contrasts with the prespecified bootstrap intervals. Full ideal values use the matrix exponential. The primary intervals form the nominal 95\% family of three; the intermediate-time rows are secondary comparisons.}
\label{q:contrasts}
\end{table}

The gauge difference was $C_0-C_g=-0.007562$, with 90\% interval $[-0.012684,\allowbreak -0.002413]$. It lies inside the chosen equivalence margin but excludes zero. Thus the margin test passed while a measurable offset remained. The absolute margin 0.02 is about $3.12|C_0^{\mathrm{ideal}}|$, with $|C_0^{\mathrm{ideal}}|=0.00641465$; it does not certify relative accuracy of this small contrast. The mean assignment error across the four available calibration matrices ranged from 1.52\% to 2.12\%; its size and the stability checks do not account for all observed evolution bias.

At the intermediate time, $C_\pi$ was resolved with the opposite sign to the ideal and $W$ was unresolved. These results and the missing short-time data limit extension of the primary observation across times. The experiment tests a phase response at finite time; a physical verification of the fourth-order short-time law requires its own temporal test.

An earlier two-instance pilot compared direct, Ruth and Strang implementations on the same device. Its 1,212,416 returned shots did not validate the four primary ideal contrasts: one sign was resolved correctly, two were opposite, one was unresolved, and no ideal value lay in its primary interval. Its changing preparation/assignment calibration and all physical records were retained. The present phase-control demonstration and that pilot are reported as separate acquisitions with distinct protocols.

Together, the circuit checks and observations establish a small hybrid implementation and a limited phase-controlled tie response, with individual biases and time dependence that remain visible in the data. Section~\ref{q:records} identifies the numerical circuit parameters and reproduction records.

\subsection{A separate temporal acquisition}
\label{sm:temporalacquisition}

We subsequently performed a temporal experiment on an auxiliary four-vertex square, with unit-cost edges $01,02,13,23$ and scale $\Lambda=2$. It is a different input from Eq.~\eqref{q:weights}. Changing the sign of coupling 13 gives
\begin{equation}
 H_0^\square=-\frac{IX+XI}{2},\qquad
 H_\pi^\square=-\frac{IX+XZ}{2}.
 \label{q:temporalH}
\end{equation}
This experiment tests transition probabilities, not a new tree-cost advantage. Its two observables and their short-time coefficients are derived in Section~\mbox{4.2}.

These two-qubit evolutions admit exact identities,
\begin{align}
 U_0^\square(\tau)&=R_{XI}(-\tau)R_{IX}(-\tau),\nonumber\\
 U_\pi^\square(\tau)&=R_{XY}(-\pi/4)
 R_{IX}(-\sqrt2\,\tau)R_{XY}(\pi/4).
 \label{q:temporalU}
\end{align}
The second identity conjugates $IX$ into $(IX+XZ)/\sqrt2$. We compiled a shared parameterized template before binding phase and time, retaining the same native pattern at zero angles. Each evolution used four CZ and ten $\sqrt X$ gates. This acquisition therefore used an exact factorization of $e^{-i\tau H_\phi^\square}$, rather than a finite-step Trotter approximation. The equality is a circuit specification; it does not imply noiseless hardware execution.

The confirmation was executed on IBM Marrakesh on 2 November 2026 using qubits 54 and 55. A separate qualification stage selected this pair using identity-transfer and calibration criteria without evaluating the quartic coefficient. Earlier trials on pair 137/147 and the qualification data were retained separately and were not pooled into confirmation. Both sources $j=0,3$, both time signs, two phases and four readout masks were measured at
\begin{equation}
 |\tau_k|=\frac{7k}{60},\qquad k=0,\ldots,12.
 \label{q:temporalgrid}
\end{equation}
The order of conditions was randomized before acquisition. Four temporal blocks were enclosed by five bare calibration packs; each pack supplied 8,192 shots per source, with 2,048 per mask. A block used the arithmetic mean of its before/after response matrices, and adjacent blocks shared a pack. Zero-time circuits with the evolution gate pattern, together with source-direction, time-sign and gauge comparisons, provided additional controls.

All nine confirmation jobs completed, returning 4,620,288 shots in 1,904 virtual histograms, of which 4,194,304 shots belonged to the main temporal series. Reported QPU usage was 1,394 s. The shot allocation, circuit family, qualification rules and endpoints were recorded before these data, with no optional stopping or extra shots for significance. All raw bit tensors and histograms were independently reconciled. Ideal probabilities reconstructed from the actual submitted circuits agreed with independent matrix exponentials to numerical precision. The subsequent statistical revision reported below used these same archived counts and acquired no additional hardware data.

\subsection{Native resources and reproduction records}
\label{q:records}

Coefficient preservation does not by itself determine a useful experimental circuit. In the independent compilation ensemble, the four-vertex direct and primary-product circuits had the same mean CZ count, about 2.94. At six to eight vertices, direct synthesis used mean counts of 29.00--29.10, while the primary products used 567.85--623.41. These counts refer to the historical offline target and the fixed compilation policy of the validation, separately from the three-CZ phase demonstration and the four-CZ exact temporal circuit. They show why logical product order and rotation counts need a native-circuit audit. Dense direct synthesis includes a classically computed target matrix, so these small circuits do not establish a scalable synthesis method or a computational advantage.

\FloatBarrier

The public code and data record is fixed at commit \texttt{68e01c78} of the accompanying repository~\cite{LuizMSTCode2026}. Within that version, \path{research/artigo_discover/manuscrito_mvc/circuitos_etapa2.json} records the seven directly synthesized evolutions: each circuit's time, phase, global phase and chronological local gate layers, including all numerical rotation parameters. These are the parameters of Fig.~\mbox{1}; Eq.~\eqref{q:native} specifies their matrix ordering. The repository's \path{REPRODUCIBILITY.md} gives the environment, data layout and commands for recovering the reported outputs.

The records retain native circuits, serialized submissions, server responses, count histograms and packed bits. Independent decoding of the 144 returned phase-campaign bit tensors reproduced all histograms; five corrected matrices agreed with an independent reconstruction to $2.3\times10^{-16}$. Client metadata alone are not an independent physical witness of server execution. The recorded phase-campaign quota changed from 4,909 to 4,739 s, consistent with the reported 170 s charge; usage entries for the failed jobs were still marked pending at the last query.

An earlier positive-cost complete-graph candidate with weak couplings was not used for the temporal confirmation. Its six-CZ Strang and twelve-CZ Ruth templates and prospective shot estimates remain historical records, separately from the executed square experiment. They are not resources or data of that acquisition.

\section{Remainder bounds, product formulas and ordering equivalence}
\label{sm:derivation}

\subsection{Remainders, degeneracies and sampling}

For a real symmetric generator $H$ with $\|H\|\le1$ and an off-diagonal transition $i\ne j$, the probability satisfies the uniform bound
\begin{equation}
 \left|P_{ij}(\tau)-\tau^2H_{ij}^2-\tau^4c_{ij}\right|
 \le C_6|\tau|^6,\qquad C_6=\frac{2^6}{6!}=\frac4{45}.
 \label{s:remainder}
\end{equation}
To see this, set $\rho_j=U\ket{j}\bra{j}U^\dagger$, so $P_{ij}=\tr[\ket{i}\bra{i}\rho_j]$. With $\operatorname{ad}_H(X)=[H,X]$, the trace norm gives
\[
 |P_{ij}^{(k)}(\tau)|\le\|\operatorname{ad}_H^k(\rho_j)\|_1
 \le(2\|H\|)^k.
\]
Taylor's theorem through fifth order, with vanishing odd derivatives at zero, gives Eq.~\eqref{s:remainder}. The absolute-row-sum scale bounds the norm also for the real signed control.

For a nonzero tied-edge coefficient gap $\gamma=|c_e-c_f|$, its sign is preserved if
\begin{equation}
 \tau^4\gamma>2C_6|\tau|^6+2\epsilon_P,
 \label{s:ordering}
\end{equation}
where $\epsilon_P$ bounds the error of each estimated symmetric edge score. In particular, noiseless nondegenerate ties have a sufficiently small window. This sufficient bound can be conservative. If $\gamma=0$, fourth order does not determine the tie, and higher orders or the residual deterministic ordering can matter. There is no uniform finite-time guarantee independent of the coefficient gaps.

Shrinking time reduces truncation error but also suppresses the measurable contrast. When individual transition probabilities scale as $\tau^2$, the variance of a contrast estimated with $S$ shots per relevant source scales as $\tau^2/S$, while its squared fourth-order signal scales as $\tau^8$. A fixed signal-to-noise ratio can thus require $S\propto\tau^{-6}$. A nonvanishing error floor can worsen this scaling to $\tau^{-8}$. Particular observables may have more favourable leading variances; these estimates describe the generic tradeoff, not a lower bound for every input.

\subsection{Which product formulas preserve the probability coefficient?}

Consider a fixed product of factors with real Hermitian matrices $K_\ell$ and real coefficients $a_\ell$,
\begin{equation}
 V(\tau)=\prod_\ell e^{-ia_\ell\tau K_\ell}
        =\sum_{k\ge0}(-i\tau)^kR_k,
 \qquad R_k\ \text{real}.
 \label{s:productseries}
\end{equation}
If the product is consistent, $R_1=H$. For $i\ne j$, its fourth-order probability coefficient is
\begin{equation}
 c_{ij}^{V}=(R_2)_{ij}^2-2H_{ij}(R_3)_{ij}.
 \label{s:digitalc4}
\end{equation}
Matching $R_2=H^2/2$ and $R_3=H^3/6$ is sufficient to recover Eq.~\textup{(13)}. Thus a third-order operator formula built from real factors preserves the probability through fourth order; the first possible probability error is sixth order, although its operator error can start at fourth order. A generic second-order formula need not preserve $c_{ij}$. Reality of the summed Hamiltonian alone is insufficient if its individual factors do not satisfy the real-matrix assumption.

These statements follow directly by multiplying series. Appending a factor $e^{-ia\tau K}$ on the right gives the reproducible recurrence
\begin{equation}
 R_k^{\mathrm{new}}=\sum_{m=0}^kR_{k-m}^{\mathrm{old}}\frac{a^mK^m}{m!}.
 \label{s:recurrence}
\end{equation}
For two groups $A,B$, the Ruth schedule~\cite{Hatano2005} used here, in matrix product order, is
\begin{equation}
 (A,\tfrac7{24}),\ (B,\tfrac23),\ (A,\tfrac34),\
 (B,-\tfrac23),\ (A,-\tfrac1{24}),\ (B,1).
 \label{s:ruth}
\end{equation}
The corresponding gates are applied in reverse order. The commuting real Pauli groups of Section~\mbox{3} meet the assumptions. Direct synthesis is checked against its complete target matrix separately: its time-dependent native angles are not the fixed linear schedule of Eq.~\eqref{s:productseries}.

\subsection{Numerical verification and the finite-time tie example}

We evaluated coefficients with exact rational arithmetic, checked Eq.~\textup{(13)} against series multiplication, and compared the resulting probability series with independent dense matrix exponentials. For the input in Eq.~\eqref{q:weights},
\begin{equation}
 a_4=-\frac{5024864156}{457371410625}
     \simeq-0.01098639757.
 \label{s:maincoefficient}
\end{equation}
Table~\ref{s:times} and Fig.~\ref{s:convergence} show convergence to this coefficient. The twelfth-order series was also compared with the full exponential, with its error checked against a four-transition Taylor remainder bound of order fourteen. These computations verify the analytical prediction numerically; they are not additional measurements.

\begin{table}[htbp]
\centering\small

\begin{tabular}{rrrr}
\toprule
$\tau$ & Full ideal $W$ & Fourth order & Relative error \\
\midrule
0.1000 & -0.00000110 & -0.00000110 & 0.11\% \\
0.3000 & -0.00008808 & -0.00008899 & 1.03\% \\
0.6727 & -0.00213645 & -0.00225001 & 5.32\% \\
1.3454 & -0.02916240 & -0.03600023 & 23.45\% \\
\bottomrule
\end{tabular}

\caption{Classical evaluation of the ideal quantum contrast for the previously measured input. Relative error is $|W_{\mathrm{full}}-a_4\tau^4|/|W_{\mathrm{full}}|$. Printed values are rounded.}
\label{s:times}
\end{table}

At the measured primary time, $a_4\tau_*^4=-0.03600023$, whereas the complete ideal evolution gives $-0.02916240$. The 23.45\% difference matters: the fourth-order truncation lies outside the primary experimental interval $[-0.03490345,-0.02503747]$, while the complete ideal value lies inside. This does not contradict the asymptotic expansion. The primary time was selected for visibility, and higher orders already contribute. Together with the mixed intermediate result and absent short-time block, it prevents interpreting that phase-control campaign as a physical verification of Eq.~\textup{(14)}.

\begin{figure}[htbp]
\centering
\includegraphics[width=\linewidth]{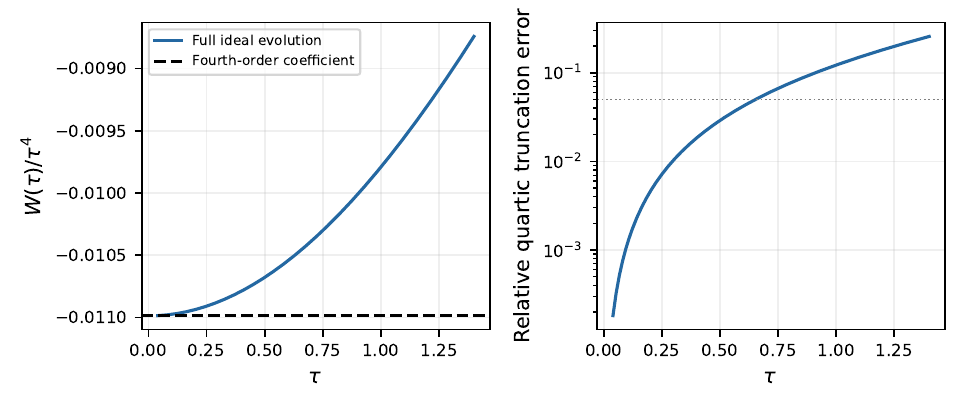}
\caption{Ideal short-time behaviour for the existing input. Left: $W(\tau)/\tau^4$ approaches the exact coefficient in Eq.~\eqref{s:maincoefficient}. Right: relative error of the fourth-order truncation. The dotted line marks 5\% error. These curves are calculated, not physical data}
\label{s:convergence}
\end{figure}

\subsection{Higher temporal orders of the auxiliary observables}

For D and E in Eqs.~\textup{(16)}--\textup{(17)}, the closed laws also determine the higher powers.

In the first phase, the two independent rotations give the opposite-vertex probability $\sin^4(\tau/2)$ and total neighbour probability $\tfrac12\sin^2\tau$. In the second phase, the two Pauli terms anticommute, giving zero opposite-vertex probability and neighbour probability $\sin^2(\tau/\sqrt2)$. Thus
\begin{align}
 D(\tau)&=\frac{\tau^4}{16}-\frac{\tau^6}{96}
             +\frac{\tau^8}{1280}+O(\tau^{10}),\nonumber\\
 E(\tau)&=-\frac{\tau^4}{12}+\frac{\tau^6}{60}
             -\frac{\tau^8}{720}+O(\tau^{10}).
 \label{s:square_series}
\end{align}
The two observables probe different contributions to Eq.~\textup{(13)}: $D$ involves paths with two steps between nonadjacent vertices, whereas $E$ cancels the quadratic term on existing edges and retains interference between one-step and three-step amplitudes. They concern the transition mechanism on an auxiliary input, separately from the tied decision in Eq.~\eqref{p:W}.

\subsection{Conditional equivalence of the classical reduction}

\begin{proposition}
For a fixed finite input with real $H$, suppose all distinct edges within each equal-cost class have distinct fourth-order coefficients. There exists $\tau_0>0$ such that, for $0<|\tau|<\tau_0$, exact finite-time scores and QW4 give the same ordered edges and hence the same accepted sequence for any common degree bound.
\end{proposition}
\begin{proof}
For each tied pair with gap $\gamma>0$, Eq.~\eqref{s:ordering} with zero estimation error guarantees its predicted sign whenever $2C_6\tau^2<\gamma$. There are finitely many such pairs, so the minimum of their positive windows works simultaneously. Different-cost edges already share the same primary order. The deterministic controller then sees the same ordered candidates and makes the same decisions at every step. If all cost classes are singletons, the secondary scores have no effect and no time restriction is required.
\end{proof}

The proposition concerns exact arithmetic and noiseless probabilities. Numerical resolution and measurement errors require their own margins. It does not determine a tie with equal fourth-order coefficients: the finite-time probability can break it through sixth or higher orders while QW4 uses its residual index keys. The admissible window depends on the graph and its coefficient gaps. This conditional reduction therefore permits the inversions and residual decisions observed below.

Both constructions retain Kruskal's MST guarantee when the degree test is omitted, because original cost remains primary. With a degree bound, QW4 remains a greedy heuristic. Its coefficient formula supplies a reproducible rule, not a certificate of minimum constrained cost. An incomplete forest must be reported even if another ordering could construct a feasible tree.

\section{Temporal inference, controls and statistical validation}
\label{sm:temporal}

After describing the datasets, we define the calibrated contrasts and decision criteria, then present the likelihood, results and validation diagnostics. The analysis retains the raw histograms and shared calibration information throughout.

\subsection{Acquisition, provenance and unchanged endpoints}

The temporal confirmation is a separate experiment from the earlier phase-controlled tree choice. It uses the auxiliary unit-cost square, two known sources, real phases 0 and $\pi$, both time signs, four masks and 13 absolute times $7k/60$, $k=0,\ldots,12$. Its shared parameterized circuit was compiled before binding time or phase. All evolutions retain four CZ and ten $\sqrt X$ gates, including the zero-time condition. The circuit factorization is exact in ideal arithmetic; the temporal probability fit is a separate approximation.

The real confirmation returned 4,620,288 shots in nine completed jobs on qubits 54 and 55, including 4,194,304 shots in the main temporal series. Four blocks share five surrounding bare calibration packs. Additional zero-time circuits retain the evolution gate pattern, and the gauge comparison is evaluated in the designated second block. Conditions were randomized before acquisition. Pair qualification used calibration-transfer and identity checks without computing the quartic coefficient. Pilots and the qualification data were not pooled with confirmation.

The Fake dataset uses the same 1,904 virtual conditions and shot allocation, with stationary channels from the archived 26 February 2025 FakeMarrakesh calibration. The sample seed was 20261121, fixed before its synthetic counts; the ideal reference sample used 20261120. These are simulations, not additional device measurements. Both datasets' archived bitstrings were reconciled with the histograms used in this revision.

Qualification requires response condition number at most 2, maximum pre/post entry difference at most 0.015 and corrected matched-zero identity error at most 0.03. Table~\ref{sm:qualification} gives the observed values using the original point-estimate rules. Passing them does not identify detector error separately or prove stationarity.

\begin{table}[htbp]
\centering\small

\begin{tabular}{llrrr}
\toprule
Dataset & Block & Condition & Pre/post change & Zero-time error \\
\midrule
Fake & 1 & 1.0605 & 0.00610 & 0.01298 \\
Fake & 2 & 1.0605 & 0.00513 & 0.00968 \\
Fake & 3 & 1.0573 & 0.00500 & 0.01448 \\
Fake & 4 & 1.0573 & 0.00281 & 0.01228 \\
IBM & 1 & 1.0156 & 0.00232 & 0.02794 \\
IBM & 2 & 1.0195 & 0.00293 & 0.02502 \\
IBM & 3 & 1.0195 & 0.00208 & 0.01912 \\
IBM & 4 & 1.0184 & 0.00293 & 0.02286 \\
\bottomrule
\end{tabular}

\caption{Qualification of the four confirmation blocks. The same rules and calibration averages are used before and after the statistical revision.}
\label{sm:qualification}
\end{table}

\FloatBarrier
\subsection{Operational contrasts and the six equivalence controls}
\label{sm:controlestimands}

Let $q_{ki}$ denote the probability of outcome $i=0,\ldots,3$ in histogram $k$, after undoing its known XOR mask. The index $k$ identifies a complete condition, including calibration or physical block, source, mask and, for physical data, time, phase, time sign and gauge. The observed count is $C_{ki}$ and the fixed stratum size is $N_k=\sum_i C_{ki}$. Substituting the saturated estimate $\widehat q^{\rm sat}_{ki}=C_{ki}/N_k$ in the following construction gives the empirical corrected contrasts; the likelihood instead varies the normalized vectors $q_k$.

Calibration pack $r=0,\ldots,4$ supplies an effective response matrix, and block $b=1,\ldots,4$ uses its adjacent-pack mean:
\begin{equation}
 A^{\mathrm{cal}}_{r,ij}
 =\frac{\sum_mN^{\mathrm{cal}}_{r,j,m}q^{\mathrm{cal}}_{r,j,m,i}}
 {\sum_mN^{\mathrm{cal}}_{r,j,m}},\qquad
 \bar A_b=\frac{A^{\mathrm{cal}}_{b-1}+A^{\mathrm{cal}}_b}{2}.
 \label{sm:response}
\end{equation}
For physical data, let $\tau\geq0$, $\sigma\in\{-1,+1\}$ denote the time sign, and $g=0,1$ denote the ordinary or gauged circuit. For a set $\mathcal S$ of time signs, average the mask/sign strata by their recorded shot sizes:
\begin{equation}
 F^{g,\mathcal S}_{b,\phi,j,i}(\tau)
 =\frac{\sum_{\sigma\in\mathcal S,m}N_{b,\tau,\phi,j,g,\sigma,m}
 q_{b,\tau,\phi,j,g,\sigma,m,i}}
 {\sum_{\sigma\in\mathcal S,m}N_{b,\tau,\phi,j,g,\sigma,m}},\qquad
 \boldsymbol p^{g,\mathcal S}_{b,\phi,j}
 =\bar A_b^{-1}\boldsymbol F^{g,\mathcal S}_{b,\phi,j}.
 \label{sm:physicalaverage}
\end{equation}
Each $q$ or $N$ with explicit condition indices denotes the corresponding histogram quantity. The corrected vectors $\boldsymbol p$ are operational quantities; their entries need not be nonnegative.

For $j=0,3$, define phase contrasts and their source mean by
\begin{align}
 d^{g,\mathcal S}_{D,b,j}(\tau)
 &=p^{g,\mathcal S}_{b,0,j,3-j}(\tau)
   -p^{g,\mathcal S}_{b,\pi,j,3-j}(\tau),\nonumber\\
 d^{g,\mathcal S}_{E,b,j}(\tau)
 &=\sum_{i=1,2}\left[p^{g,\mathcal S}_{b,0,j,i}(\tau)
   -p^{g,\mathcal S}_{b,\pi,j,i}(\tau)\right],\nonumber\\
 y^{g,\mathcal S}_{O,b}(\tau)
 &=\tfrac12\left[d^{g,\mathcal S}_{O,b,0}(\tau)
                  +d^{g,\mathcal S}_{O,b,3}(\tau)\right],\qquad O\in\{D,E\}.
 \label{sm:contrasts}
\end{align}
The temporal fit uses $y_b=y^{0,\{-,+\}}_{O,b}$ at the 52 block/time points. The two explicit signs in $\{-,+\}$ stand for $\{-1,+1\}$; a singleton retains only that time direction.

All equivalence controls are evaluated at $\tau_c=1.4$. The gauge is $Q=I\otimes Z=\diag(1,-1,1,-1)$, applied as $H_\phi\mapsto QH_\phi Q$ in both phases. With every term below evaluated at $\tau_c$, the signed control differences are
\begin{align}
 G_O&=y^{1,\{-,+\}}_{O,2}-y^{0,\{-,+\}}_{O,2},\nonumber\\
 S_O&=\tfrac14\sum_{b=1}^4
       \left[d^{0,\{-,+\}}_{O,b,0}-d^{0,\{-,+\}}_{O,b,3}\right],\nonumber\\
 T_O&=\tfrac14\sum_{b=1}^4
       \left[y^{0,\{+\}}_{O,b}-y^{0,\{-\}}_{O,b}\right].
 \label{sm:controlvalues}
\end{align}
Thus the gauge comparison uses block 2, whereas source-direction and time-sign controls average all four blocks. These formulas specify the six controls (three for each observable), including their orientations. All have ideal value zero. The earlier phase-campaign gauge difference $C_0-C_g$ is a different observable with its own orientation.

\subsection{Unchanged decision criteria}

For each dataset and observable, the primary temporal model has block-specific powers 0 and 2 and common powers 4, 6 and 8, as in Eq.~\textup{(18)}. Its coefficient test requires sign resolution in the predicted direction, inclusion of the analytical $a_4$ in the interval $[l,u]$, relative half-width $(u-l)/(2|a_4|)\leq0.5$, and goodness-of-fit probability at least 0.025. The denominator is the analytical coefficient magnitude, not the fitted magnitude. Here $a_{4,D}=1/16$ and $a_{4,E}=-1/12$.

The six controls in Eq.~\eqref{sm:controlvalues} use nominal individual 99.1667\% intervals. Each interval must lie strictly inside $(-\delta_O,\delta_O)$, where $\delta_O=0.1|a_{4,O}|\tau_c^4$: 0.02401000 for D and 0.03201333 for E.

For model coefficients $\beta$, define the conservative nonquartic-to-quartic ratio
\begin{align}
 R(\beta)&=\max_{\substack{b=1,\ldots,4\\
 \tau\in\{0.7,0.816667,0.933333\}}}
 \frac{|b_{0,b}|+|b_{2,b}|\tau^2+|b_6|\tau^6+|b_8|\tau^8}
 {|b_4|\tau^4},\nonumber\\
 R_{.95}&=Q_{.95}\!\left[R(\beta^*)\right].
 \label{sm:dominance}
\end{align}
Here $\beta^*$ is a refitted parametric-bootstrap replicate, and $Q_{.95}$ is its 95\% quantile. Absolute terms prevent cancellations from hiding nonquartic contributions. The ratio concerns the retained terms, without bounding an omitted temporal remainder or an unmodeled systematic quartic contribution. Its bootstrap quantile is not an exact physical 95\% confidence bound. The dominance endpoint requires $R_{.95}<0.5$; it is distinct from coefficient inclusion and is not an observed event fraction. The revised analysis changes the inferential method while preserving these windows, powers, margins, thresholds and observations. It is a post-acquisition reanalysis, not a new independent confirmation.

\subsection{Joint multinomial likelihood}

The original empirical-variance treatment can assign zero estimated variance to a rare event with zero observed counts. A count likelihood retains uncertainty about that event instead of treating its empirical frequency as a known probability. Each histogram defined above remains a separate four-category multinomial, even when it shares time or source labels with another. The likelihood does not assume equality across masks, sources, time signs or gauge conditions:
\begin{equation}
 \ell(\boldsymbol q)=\sum_{k,i} C_{ki}\log q_{ki},\qquad
 \sum_iq_{ki}=1,\quad q_{ki}\ge0.
 \label{sm:likelihood}
\end{equation}
The multinomial combinatorial factors cancel in likelihood ratios. A zero count contributes no $C_{ki}\log q_{ki}$ term; it does not imply that $q_{ki}$ is known to be zero.

Equations~\eqref{sm:response}--\eqref{sm:contrasts} specify the full map from raw probabilities to mask/sign averages, calibration inversion and operational contrasts. Both physical and calibration probabilities are likelihood parameters. Because a pack appears only once in Eq.~\eqref{sm:likelihood}, sharing it across adjacent blocks does not count its calibration information twice.

The fit constrains the resulting operational contrast at each of the 52 block/time points to the polynomial in the main text. All block intercepts, quadratic coefficients and common fourth-, sixth- and eighth-order coefficients are free. No unitary coefficient or control symmetry is imposed. The raw $\boldsymbol q_k$ remain valid probabilities, but inverse-corrected entries are operational estimators; they are not reinterpreted as integer binomial counts or required to form a latent unitary probability matrix. The correction still requires calibration transfer and does not remove arbitrary gate errors.

For one observable, write the polynomial design as $X\beta$ with 11 columns. An orthonormal basis $Z$ for the left null space of $X$ expresses the constraint as $Zy(\boldsymbol q)=0$, with 41 components. Fixing $b_4$ for a profile removes that column and includes its contribution in a known offset, leaving 42 components. Controls are profiled separately against their target differences, without imposing the polynomial or symmetry on their saturated alternatives.

\subsubsection{Numerical maximization and zero-count checks}

Each step linearizes the calibration-dependent constraints and solves the resulting multinomial problem on the simplexes through a convex dual. For a linear constraint $Jq=r$ and dual multiplier $\lambda$, the interior solution for each histogram is
\begin{equation}
 q_{ki}=\frac{C_{ki}}{\nu_k+(J^T\lambda)_{ki}},
\end{equation}
where $\nu_k$ normalizes its probabilities. Bisection solves the normalization and Newton updates the dual. Repeating the linearization closes the nonlinear calibration constraints.

Numerically, a log barrier $\mu\sum_{ki}\log q_{ki}$, with $\mu=10^{-5}$, keeps the operations in the interior; the numerator in that subproblem becomes $C_{ki}+\mu$. This is a numerical regularization of the optimization, not a variance estimate or a change to the observed counts. Reported objectives use the original $C_{ki}$. Fits at $\mu=10^{-3}$ and $10^{-7}$ changed the reported coefficients by less than $3.4\times10^{-7}$.

Jacobian checks included both physical and calibration probabilities. Distinct initial calibration matrices converged to the same coefficients, and a constrained multinomial example with an analytic solution was recovered. A second analytic check used counts $(0,100,0,0)$ and $(10,90,0,0)$ with a common first-category probability: the exact maximum is $(0.05,0.95,0,0)$ for both histograms. The implementation recovered it despite the first histogram's zero event count. Across the main bootstrap, the largest recorded nonlinear constraint residual was $1.59\times10^{-9}$. These checks audit numerical implementation; they prove neither global optimality of the nonlinear likelihood nor universal statistical coverage.

\subsubsection{Profiles, deviance and resampling}

Let $\mathcal M$ be the polynomial model: normalized raw probability vectors with invertible response matrices for which the 52 operational contrasts satisfy $y(q)=X\beta$. Let $\mathcal M_a$ impose the additional restriction $b_4=a$. Distinguish the saturated estimate $\widehat q_{\rm sat}$ from the model and restricted likelihood optima,
\begin{equation}
 \widehat q_{\mathcal M}\in\arg\max_{q\in\mathcal M}\ell(q),\qquad
 \widehat q_{\mathcal M,a}\in\arg\max_{q\in\mathcal M_a}\ell(q).
 \label{sm:likelihoodoptima}
\end{equation}
The numerical procedure above approximates these optima. For a candidate coefficient $a$, the profile statistic uses the model optimum as its reference:
\begin{equation}
 T(a)=2\left[\ell(\widehat q_{\mathcal M})-\ell(\widehat q_{\mathcal M,a})\right].
 \label{sm:profileLR}
\end{equation}
Thus $T(\widehat b_4)=0$ at a common optimum. This differs from the goodness-of-fit deviance,
\begin{equation}
 G=2\left[\ell(\widehat q_{\rm sat})-\ell(\widehat q_{\mathcal M})\right].
 \label{sm:deviance}
\end{equation}
The coefficient-profile implementation subtracts the negative log likelihood of the model fit; the saturated likelihood is used for $G$. The six controls instead have saturated alternatives without the polynomial restriction: each control profile compares that saturated optimum with its own constrained fit. The nominal individual 97.5\% coefficient interval uses the $\chi^2_1$ cutoff, whose large-sample basis requires regularity~\cite{Wilks1938}. The finite-sample checks below are therefore material, especially near probability boundaries.

Goodness of fit uses deviance against the saturated multinomial model. We calibrate it with 2,000 parametric bootstrap replicates per dataset/observable, drawing all physical and calibration histograms at their recorded stratum sizes from the fitted model and refitting the shared parameters~\cite{Davison1997}. The same replicates supply coefficient diagnostics and the dominance quantile. Bootstrap seed 20261210 is combined deterministically with dataset, observable and replicate indices; no seed is selected by outcome. The Monte Carlo intervals for the dominance quantile quantify simulation precision rather than an additional physical confidence interval. Validation of coefficient profiles does not by itself guarantee the repeated-sampling coverage of this nonlinear dominance procedure.

Table~\ref{sm:dominanceMC} reports that simulation precision using the existing replicates. Let $r_{(1)}\leq\cdots\leq r_{(B)}$ be their ordered dominance ratios. The implemented 95\% Monte Carlo interval is $[r_{(L)},r_{(U)}]$, with
\begin{align*}
 L&=\max\{1,F^{-1}_{B,0.95}(0.025)\},\\
 U&=\min\{B,F^{-1}_{B,0.95}(0.975)+1\}.
\end{align*}
Here $F^{-1}_{B,0.95}$ is the binomial quantile function; $B=2000$ gives $L=1880$ and $U=1920$. The point estimate uses linear interpolation of the empirical 0.95 quantile. All four intervals remain on the same side of 0.5 as their registered point decisions; this checks Monte Carlo precision, without certifying the physical coverage of $R_{.95}$.

\begin{table}[htbp]
\centering\small

\begin{tabular}{llrr}
\toprule
Dataset & Obs. & $R_{.95}$ & MC interval (95\%) \\
\midrule
Fake & D & 0.45266 & $[0.44127, 0.47330]$ \\
Fake & E & 0.59130 & $[0.58099, 0.59924]$ \\
IBM & D & 0.47558 & $[0.46868, 0.48450]$ \\
IBM & E & 0.82647 & $[0.78379, 0.91373]$ \\
\bottomrule
\end{tabular}

\caption{Precision of the dominance quantile from the 2,000 archived bootstrap replicates per row. These intervals quantify Monte Carlo uncertainty and are not additional physical confidence intervals. No new replicates were generated for this summary.}
\label{sm:dominanceMC}
\end{table}

\FloatBarrier
\subsection{Original and revised results}

Table~\ref{sm:original} and Fig.~\ref{sm:profiles} compare the original bootstrap results with the revised profiles. Table~\ref{sm:decisions} separates coefficient and dominance decisions; Table~\ref{sm:controls} gives all six controls defined in Eq.~\eqref{sm:controlvalues}. In the original IBM analysis, both coefficient tests passed, but E dominance and the time-sign equivalence control of D did not. The revision leaves the same two overall limitations. It does not retroactively convert the original campaign into a fully successful confirmation.

\begin{table}[htbp]
\centering\small

\begin{tabular}{lllrr}
\toprule
Dataset & Obs. & Analysis & Estimate & Nominal 97.5\% interval \\
\midrule
Fake & D & Original bootstrap & 0.05594 & $[0.04418, 0.06773]$ \\
Fake & D & Count likelihood & 0.05613 & $[0.04424, 0.06801]$ \\
Fake & E & Original bootstrap & -0.09279 & $[-0.12888, -0.05775]$ \\
Fake & E & Count likelihood & -0.09226 & $[-0.12846, -0.05606]$ \\
IBM & D & Original bootstrap & 0.05871 & $[0.04746, 0.06982]$ \\
IBM & D & Count likelihood & 0.05942 & $[0.04815, 0.07069]$ \\
IBM & E & Original bootstrap & -0.07685 & $[-0.11141, -0.04317]$ \\
IBM & E & Count likelihood & -0.07620 & $[-0.10973, -0.04268]$ \\
\bottomrule
\end{tabular}

\caption{Original and revised coefficient inference on identical archived counts. Original intervals use the empirical bootstrap/GLS analysis; revised intervals use count likelihood profiles. Both are nominal individual 97.5\% intervals.}
\label{sm:original}
\end{table}

\begin{table}[htbp]
\centering\small

\begin{tabular}{lrrrrr}
\toprule
Dataset & Qualification & Coefficients & Controls & Dominance D & Dominance E \\
\midrule
Fake & 4/4 & 2/2 & 6/6 & Pass & Unresolved \\
IBM & 4/4 & 2/2 & 5/6 & Pass & Unresolved \\
\bottomrule
\end{tabular}

\caption{Revised decisions with unchanged physical criteria. Unresolved indicates failure to establish the required dominance bound, not proof of its physical absence.}
\label{sm:decisions}
\end{table}

\begin{figure}[htbp]
\centering
\includegraphics[width=\linewidth]{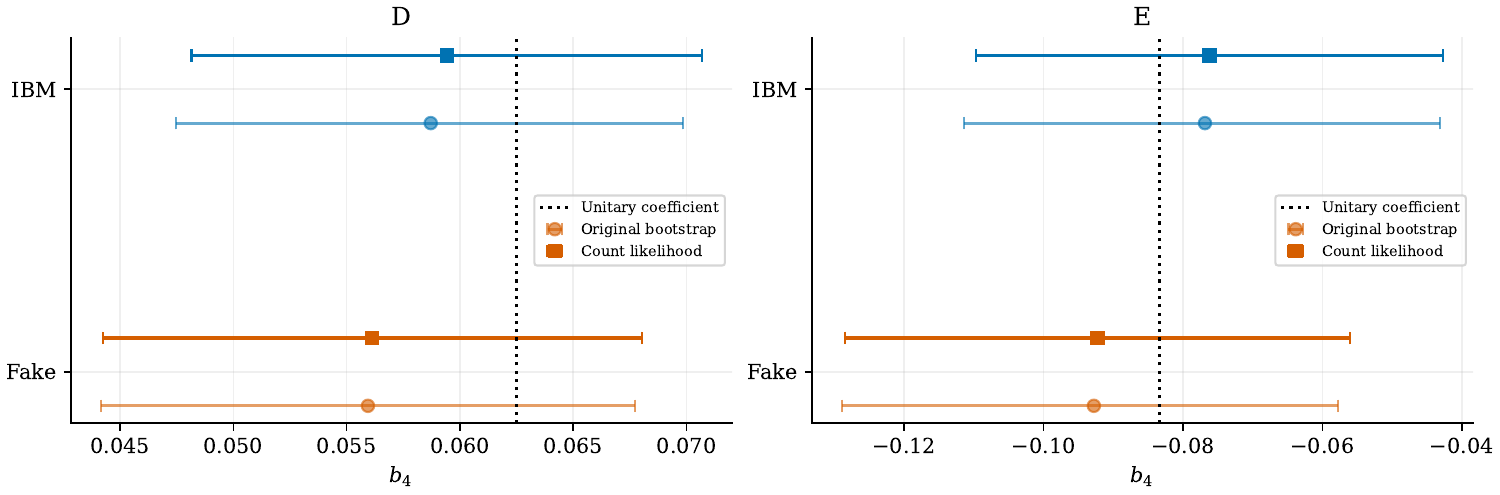}
\caption{Coefficient intervals before and after the statistical revision. The target unitary coefficient is marked by a dotted line. Original and revised estimates use the same counts; the evolution and Hamiltonian were not changed}
\label{sm:profiles}
\end{figure}

\begin{table}[htbp]
\centering\scriptsize

\begin{tabular}{lllrrr}
\toprule
Dataset & Obs. & Control & Estimate & Profile interval (99.1667\%) & Margin \\
\midrule
Fake & D & Gauge & -0.00207 & $[-0.00957, 0.00539]$ & 0.02401 \\
Fake & D & Source direction & 0.00252 & $[-0.00369, 0.00874]$ & 0.02401 \\
Fake & D & Time sign & -0.00167 & $[-0.00785, 0.00451]$ & 0.02401 \\
Fake & E & Gauge & -0.00308 & $[-0.01623, 0.01006]$ & 0.03201 \\
Fake & E & Source direction & 0.00343 & $[-0.00741, 0.01428]$ & 0.03201 \\
Fake & E & Time sign & -0.00608 & $[-0.01691, 0.00476]$ & 0.03201 \\
IBM & D & Gauge & 0.00368 & $[-0.00348, 0.01080]$ & 0.02401 \\
IBM & D & Source direction & 0.00600 & $[0.00008, 0.01191]$ & 0.02401 \\
IBM & D & Time sign & 0.01888 & $[0.01298, 0.02478]$ & 0.02401 \\
IBM & E & Gauge & 0.00103 & $[-0.01158, 0.01363]$ & 0.03201 \\
IBM & E & Source direction & -0.00135 & $[-0.01174, 0.00905]$ & 0.03201 \\
IBM & E & Time sign & 0.01061 & $[0.00021, 0.02100]$ & 0.03201 \\
\bottomrule
\end{tabular}

\caption{Revised controls of Eq.~\eqref{sm:controlvalues}. Each interval must lie strictly within its signed margin. Only IBM time-sign D fails; its point estimate remains inside.}
\label{sm:controls}
\end{table}

\FloatBarrier
\subsection{Finite-sample validation}

The bootstrap includes 8,000 synthetic experiments, with 2,000 per dataset and observable. In the first 500 of each set, we also evaluate the likelihood ratio at the generating model's known coefficient, yielding Table~\ref{sm:coveragefitted}.

\begin{table}[htbp]
\centering\small

\begin{tabular}{llrr}
\toprule
Generator & Obs. & Inclusion (\%) & MC interval (95\%, \%) \\
\midrule
Fake & D & 97.0 & $[95.1, 98.3]$ \\
Fake & E & 97.8 & $[96.1, 98.9]$ \\
IBM & D & 97.8 & $[96.1, 98.9]$ \\
IBM & E & 97.8 & $[96.1, 98.9]$ \\
\bottomrule
\end{tabular}

\caption{Profile inclusion at the known coefficient of each fitted generator, 500 replicates per row. Nominal individual coverage is 97.5\%; Monte Carlo intervals are exact binomial 95\% intervals.}
\label{sm:coveragefitted}
\end{table}

Another 500 experiments per observable use the exact archived Fake probabilities with independent seed 20261211. Their generator is not forced to be an eighth-order polynomial. All 9,000 experiments completed without numerical failures or discarded replicates.

For this exact-Fake validation, the population target uses the same design and polynomial constraints as the count likelihood. Let $p^{\rm F}_{ki}$ be the known, unmasked Fake probabilities, including the physical and calibration histograms. Replacing the counts by $N_kp^{\rm F}_{ki}$ defines the expected-likelihood projection,
\begin{equation}
 q^\star\in\arg\max_{q\in\mathcal M}
 \sum_{k,i}N_kp^{\rm F}_{ki}\log q_{ki},\qquad b_4^\star=b_4(q^\star).
 \label{sm:populationtarget}
\end{equation}
The $N_k$ retain the recorded shot allocation, and calibration probabilities are varied jointly through the same response and contrast maps. This criterion is equivalent to minimizing the shot-weighted sum of multinomial Kullback--Leibler divergences; it is not a least-squares fit to corrected contrasts. The archived numerical fits use the same $\mu=10^{-5}$ barrier and stopping rules, approximating Eq.~\eqref{sm:populationtarget} without proving a unique global optimum. Their targets are $b_{4,D}^\star\simeq0.0616041331$ and $b_{4,E}^\star\simeq-0.0821361178$, recorded in \path{POPULACAO_FAKE.json} within the repository's joint-analysis directory. These are the population-projection targets in the following validation and the population-fit column of Table~\ref{sm:channel}. They differ from both the channel Taylor coefficients and the targets of Gaussian precision forecasts.

\begin{table}[htbp]
\centering\small

\begin{tabular}{llrr}
\toprule
Obs. & Target & Inclusion (\%) & MC interval (95\%, \%) \\
\midrule
D & Population projection & 96.8 & $[94.9, 98.2]$ \\
D & Channel Taylor coefficient & 96.6 & $[94.6, 98.0]$ \\
D & Unitary Taylor coefficient & 96.0 & $[93.9, 97.5]$ \\
E & Population projection & 97.6 & $[95.8, 98.8]$ \\
E & Channel Taylor coefficient & 97.6 & $[95.8, 98.8]$ \\
E & Unitary Taylor coefficient & 98.0 & $[96.4, 99.0]$ \\
\bottomrule
\end{tabular}

\caption{Independent exact-Fake validation, 500 experiments per observable. The population projection is the expected-likelihood target in Eq.~\eqref{sm:populationtarget}. Channel and unitary Taylor values are different comparison targets; inclusion of the unitary value is not coverage of a true coefficient of the noisy channel.}
\end{table}

\begin{table}[htbp]
\centering\small

\begin{tabular}{llrr}
\toprule
Obs. & Control & Inclusion (\%) & MC interval (95\%, \%) \\
\midrule
D & Gauge & 99.4 & $[98.3, 99.9]$ \\
D & Source direction & 99.0 & $[97.7, 99.7]$ \\
D & Time sign & 98.8 & $[97.4, 99.6]$ \\
E & Gauge & 99.0 & $[97.7, 99.7]$ \\
E & Source direction & 99.2 & $[98.0, 99.8]$ \\
E & Time sign & 99.2 & $[98.0, 99.8]$ \\
\bottomrule
\end{tabular}

\caption{Inclusion of the true Fake control differences in nominal individual 99.1667\% profiles. Each row uses 500 experiments; the generator retains its actual control differences rather than imposing zero.}
\end{table}

\FloatBarrier
An additional audit evaluates E on the same 500 experiments used in the D validation, allowing joint inclusion to be counted for both coefficient profiles and all six controls (Table~\ref{sm:coveragefamily}). It reuses these experiments and is not another independent set of 500 counts. This supplementary audit was specified after the initial marginal validation, without selecting samples or changing thresholds.

\begin{table}[htbp]
\centering\small

\begin{tabular}{lrr}
\toprule
Family & Inclusion (\%) & MC interval (95\%, \%) \\
\midrule
Two coefficients & 93.4 & $[90.86, 95.41]$ \\
Six controls & 95.0 & $[92.71, 96.74]$ \\
\bottomrule
\end{tabular}

\caption{Simultaneous inclusion in the same 500 Fake experiments. The intended lower bound is 95\% for each family.}
\label{sm:coveragefamily}
\end{table}

Joint coefficient inclusion was 93.4\%, below the 95\% point target, with Monte Carlo interval $[90.86,95.41]\%$ still containing it. Joint control inclusion was 95.0\%. These checks support only a conditional, approximate calibration at the available Monte Carlo resolution; they do not establish exact finite-sample family error control. They also do not exclude unknown drift or dependence within real hardware shots.

\FloatBarrier
\subsection{Noise, expected counts and finite-window coefficients}
\label{sm:noisediagnostics}

The earlier controlled diagnostic used fractional expected histograms $N_kq_k$ and the same finite budget in both models. These histograms remove a particular sampling realization from the central estimate, while the analysis retains sampling uncertainty at that budget. They are not the limit of infinite-shot precision. Table~\ref{sm:expectation} records those diagnostic decisions and the original single sampled realizations, all using the same original analysis. They are not substituted for the revised real/Fake results.

\begin{table}[htbp]
\centering\small

\begin{tabular}{llrr}
\toprule
Model & Counts & $E$: original upper 95\% fraction & Below 0.5 \\
\midrule
Ideal & Expected & 0.464 & Yes \\
Ideal & Sampled & 0.540 & No \\
Fake & Expected & 0.595 & No \\
Fake & Sampled & 0.592 & No \\
\bottomrule
\end{tabular}

\caption{Earlier controlled comparison of E with the original common analysis, unchanged shot allocation and criterion. Expected histograms are diagnostics rather than measured counts.}
\label{sm:expectation}
\end{table}

The ideal expected diagnostic passed, while its finite sampled realization failed E dominance. Its original D interval also failed coefficient inclusion because of a rare-count variance problem. A separate ideal binomial reference treatment recovered both analytical coefficients on those same counts, but did not automatically establish all dominance criteria. The ideal sampled experiment must therefore not be described as having passed the entire protocol after that coefficient revision.

A second diagnostic uses known simulator probabilities and first-order multinomial covariance propagation, including the shared finite calibration packs. This is a Gaussian precision forecast rather than the joint-likelihood analysis of observed counts. It gives Table~\ref{sm:knowncov}: the population nonquartic fractions remain close, while the finite-budget upper fraction grows under the noise model.

\begin{table}[htbp]
\centering\small

\begin{tabular}{lrrr}
\toprule
Model & Population fraction & Finite-budget upper fraction & Approx. s.e. of $b_4$ \\
\midrule
Ideal & 0.181 & 0.462 & 0.01340 \\
Fake & 0.181 & 0.591 & 0.01612 \\
\bottomrule
\end{tabular}

\caption{Known-probability diagnostic for E. The finite-budget column is an approximate Gaussian forecast with all calibration and physics uncertainty retained, not a new hardware measurement or the count-likelihood bootstrap.}
\label{sm:knowncov}
\end{table}

These controlled calculations support the statement that modeled noise reduces the precision for certifying quartic dominance at the adopted budget. They do not imply that the mean Fake contrast has lost dominance in the window. The real-device limitation is consistent with this explanation, without isolating a unique physical error source. The historical Fake calibration is not a complete model of the device at acquisition.

Direct propagation of Taylor series through the native gates and noise channels, including the same preparation, masks and operational assignment correction, gives the coefficients in Table~\ref{sm:channel}. These coefficients were checked against independent channel propagation and, in the ideal case, exact rational series. The eighth-order expected-likelihood population fit in Eq.~\eqref{sm:populationtarget} is a different quantity from the Taylor derivative. This distinction also applies to the ideal laws D/E, whose expansions contain powers beyond order eight.

\begin{table}[htbp]
\centering\small

\begin{tabular}{lrrr}
\toprule
Obs. & Unitary Taylor & Fake Taylor & Fake population fit \\
\midrule
D & 0.06250000 & 0.06172621 & 0.06160413 \\
E & -0.08333333 & -0.08237468 & -0.08213612 \\
\bottomrule
\end{tabular}

\caption{Taylor coefficients and the finite-window population fit. Channel and truncation biases persist even as counting uncertainty vanishes.}
\label{sm:channel}
\end{table}

\FloatBarrier
\subsection{Temporal sensitivity and scope of further work}

Table~\ref{sm:sensitivity} retains the planned sixth- and tenth-order fits and the eighth-order fit without the largest time. Their profile intervals quantify the changed precision. In the Fake D sample, the sixth-order and omitted-largest-time intervals exclude $1/16$; the sixth-order interval includes the channel Taylor coefficient. Higher-order intervals can be very wide. Consequently, the main fit's coefficient compatibility cannot be presented as robustness to all temporal choices. None of these alternatives was selected to replace the main model.

\begin{table}[htbp]
\centering\small

\begin{tabular}{lllrr}
\toprule
Dataset & Obs. & Temporal model & Estimate & Profile interval (97.5\%) \\
\midrule
Fake & D & Order 8 (main) & 0.05613 & $[0.04424, 0.06801]$ \\
Fake & D & Order 6 & 0.05745 & $[0.05278, 0.06212]$ \\
Fake & D & Order 10 & 0.00847 & $[-0.04658, 0.06369]$ \\
Fake & D & Order 8, largest time removed & 0.03337 & $[0.00540, 0.06127]$ \\
Fake & E & Order 8 (main) & -0.09226 & $[-0.12846, -0.05606]$ \\
Fake & E & Order 6 & -0.07630 & $[-0.09123, -0.06137]$ \\
Fake & E & Order 10 & -0.07699 & $[-0.25690, 0.10289]$ \\
Fake & E & Order 8, largest time removed & -0.06546 & $[-0.14608, 0.01513]$ \\
IBM & D & Order 8 (main) & 0.05942 & $[0.04815, 0.07069]$ \\
IBM & D & Order 6 & 0.05935 & $[0.05491, 0.06379]$ \\
IBM & D & Order 10 & 0.10176 & $[0.04929, 0.15450]$ \\
IBM & D & Order 8, largest time removed & 0.07824 & $[0.05186, 0.10454]$ \\
IBM & E & Order 8 (main) & -0.07620 & $[-0.10973, -0.04268]$ \\
IBM & E & Order 6 & -0.07262 & $[-0.08630, -0.05894]$ \\
IBM & E & Order 10 & -0.14196 & $[-0.31165, 0.02769]$ \\
IBM & E & Order 8, largest time removed & -0.09931 & $[-0.17556, -0.02310]$ \\
\bottomrule
\end{tabular}

\caption{Temporal sensitivity profiles on the same counts. Intervals are nominal individual 97.5\%; the main model remains order eight on the full grid.}
\label{sm:sensitivity}
\end{table}

The records support coefficient compatibility, a passed D dominance endpoint, an unresolved E endpoint and the retained hardware equivalence limitation. Full leading-law certification, exact finite-sample error guarantees or identification of an exclusive device error mechanism would require stronger evidence. This distinction limits the conclusions; the decision thresholds remain unchanged.

The reproduction records preserve compiled circuits, raw counts, bitstrings, original and revised analyses, all simulation indices, seeds, numerical audits and immutable versions. The configurations distinguish the original experimental protocol from the post-acquisition statistical revision and its conditional validation.

\section{Classical conventions, validation and finite samples}
\label{sm:classical}

Algorithm~\ref{c:algorithm} replaces the quantum acquisition of Algorithm~\ref{p:algorithm} with the analytical coefficient. The acceptance loop and original cost priority are unchanged.

\begin{algorithm}[htbp]
\caption{QW4: classical fourth-order tie breaker}
\label{c:algorithm}
\begin{algorithmic}[1]
\Require $G,w,\Delta$, with original equality classes of costs
\State Compute $B,\Lambda_G,M_2,M_3$ and cache the edge coefficients in Eq.~\textup{(19)}
\State Sort edges by $(w_{ij},-c_{ij},i,j)$, using the declared numerical convention
\State Initialize an empty forest, degree counters and disjoint components
\State Apply the acceptance loop of Algorithm~\ref{p:algorithm} to this order
\State \Return the completed tree, or an incomplete forest with its failure flag
\end{algorithmic}
\end{algorithm}

\FloatBarrier
\subsection{Preprocessing, queries and numerical conventions}

With conventional dense matrix multiplication and fixed-precision scalar arithmetic, Eq.~\textup{(19)} requires $O(N^3)$ preprocessing and $O(N^2)$ matrix storage. Caching the coefficients gives $O(1)$ scalar access for each edge. For $M=|E|$, sorting costs $O(M\log M)$; union by size with path compression and constant-time degree checks give $O(M\alpha(N))$ amortized forest operations, where $\alpha$ is the inverse Ackermann function~\cite{Tarjan1975}. Thus the supplied dense implementation costs
\begin{equation}
 O\!\left(N^3+M\log M+M\alpha(N)\right)
 \label{c:complexity}
\end{equation}
with $O(N^2+M)$ storage. Reusing a cached order for another degree bound requires only another forest pass. The constant-time statement applies to a precomputed entry, not to the full algorithm. The coefficient is unnecessary when all costs are distinct; cost-only ordering then gives the same construction.

The released implementation also supports rational costs and exact coefficient arithmetic. Its cubic count describes scalar operations; rational numerator/denominator growth and precision needed to resolve small gaps introduce additional bit costs. We report analytical operation counts here, rather than an empirical wall-clock scaling benchmark. These costs include classical preprocessing, which is also relevant when a small circuit is synthesized from a classically computed propagator.

For the previously sealed validation, integer weight keys defined quantized cost classes exactly. QW4 coefficients on these inputs were evaluated as fractions. A fixed transitive score resolution $\varepsilon_c=10^{-10}$ was applied within a class $\mathcal E_w$:
\begin{equation}
 r_e=\operatorname{round}\!\left(
 \frac{c_e-\min_{f\in\mathcal E_w}c_f}{\varepsilon_c}\right).
 \label{c:bins}
\end{equation}
Edges were sorted by descending $r_e$ and then by indices. The ideal finite-time rule used the same convention with $c_e$ replaced by $s_e$ and resolution $\varepsilon_c\tau^4$. Rounding used ties to even. Distinct original costs were never merged by this tolerance. Other diagnostic controls used their recorded floating-point scores with index residual ties. Exact coefficient degeneracies were excluded only from the mechanism-sign metric; their graphs remained in all tree metrics.

We reconstructed the classical reduction with an independent implementation that forms $B^2,B^3$ before normalizing, whereas the earlier implementation multiplied the normalized $H$. All 240 quantized coefficient vectors agreed exactly with their sealed fractions. The new implementation reproduced 2,880 accepted sequences: QW4, the full propagator and the primary product for every input at degree bounds 2, 3 and the unrestricted control. Fixed-precision coefficient evaluation also reproduced these QW4 sequences at the registered resolution. This checks the reduction against the stored data without regenerating graphs or tuning parameters.

\subsection{Validation ensemble, controls and evaluation}
\label{c:validationensemble}

The validation policy was fixed and its source hashes sealed before generating 80 new latent base graphs, with ten bases per size and family at $N=4,6,7,8$. The families used independent uniform edge latents or Euclidean distances between vertices uniformly placed in the unit square, divided by $\sqrt2$. For latent $u\in[0,1]$, continuous costs were $1+19u$; quantized versions used
\begin{equation}
 k=\min(L-1,\lfloor Lu\rfloor),\qquad
 w=1+\frac{19k}{L-1},\qquad L\in\{4,8,20\}.
 \label{c:weights}
\end{equation}
Four versions per base gave 320 complete positive-cost inputs and 640 cases at $\Delta=2,3$, plus unrestricted controls. Their graph-seed prefix was $(20261029,931)$, with no latent overlap with calibration. Weight versions and degree bounds from one base are dependent observations. The selected physical example in Section~\mbox{3} and the auxiliary temporal input in Section~\mbox{4} are separate from this validation ensemble.

The continuous reference uses the complete matrix exponential $e^{-i\tau H_0}$; its primary product uses the registered finite schedule at the same time. For four vertices, this time was $\tau=1$, with one Ruth step as in Eq.~\eqref{s:ruth} for two real commuting groups, or the exact group product if only one group was present. For six, seven and eight vertices, the primary time was $0.29286445646252357$, with one five-block fourth-order composition,
\begin{equation}
 S_4(\tau)=S_2(p\tau)^2S_2((1-4p)\tau)S_2(p\tau)^2,
 \qquad p=\frac1{4-4^{1/3}}.
 \label{c:fiveblock}
\end{equation}
Here $S_2$ is the symmetric second-order product of the real commuting Pauli groups: half steps through groups $1,\ldots,L-1$, a full step for group $L$, then the half steps in reverse order, generalizing Eq.~\eqref{q:strang}. Products act from right to left. The lower-time direct/three-block products and a one-step Strang control were retained as prespecified diagnostics. None of these parameters was retuned on validation costs or signs.

The classical comparisons retained the same primary cost and forest checks. Stable order uses vertex indices inside a cost class. Define inverse-cost strength $a_i=\sum_j B_{ij}/\Lambda_G$; the manuscript's strength control uses secondary score $-(a_i+a_j)$, preferring smaller summed strength. The two ablations use $(H^2)_{ij}^2/4$ and $-H_{ij}(H^3)_{ij}/3$ individually. Random-one uses a single random tie order. Random-best-of-30 returns the least costly of 30 independently ordered constructions, with the budget of 30 constructions explicitly retained. It does not use the optimum to choose a candidate. It is a larger-search-budget comparator, not a measurement replicate of QW4.

For evaluation only, all labelled trees were enumerated through Pr\"ufer codes~\cite{Prufer1918} to obtain the optimum $\mathcal C^*_{G,\Delta}$. The reported optimality gap is
\begin{equation}
 g=100\left(\frac{\mathcal C(T)}{\mathcal C^*_{G,\Delta}}-1\right)\%.
 \label{c:gap}
\end{equation}
The optimum is absent from the selection rule. All recorded constructions on these complete inputs returned feasible trees; general-input failure reporting remains part of the implementation. In the four-vertex cases, $\Delta=3$ is already unrestricted, which is reflected in the separate degree results below.

Paired gap differences were evaluated on the same input and degree. The registered percentile intervals use 10,000 resamples of latent bases, stratified by size and family, keeping all four weight versions and both bounds together. They are nominal 95\% intervals for the specified comparisons, not simultaneous guarantees across every comparison. We independently reconstructed these intervals from the raw tables using the registered seeds.

\subsection{Tree agreement and the limits of finite-time equivalence}

QW4 and the primary full propagator returned the same trees in all 640 bounded-degree cases. The primary product returned the same cost in all cases and the same tree in 603. Its accepted sequence agreed in 600. Table~\ref{c:agreement} separates these endpoints and the two degree bounds.

\begin{table}[htbp]
\centering\small

\begin{tabular}{rrrrrr}
\toprule
$N$ & QW4 gap, $\Delta=2$ & QW4 gap, $\Delta=3$ & Cost & Tree & Sequence \\
\midrule
4 & 1.453 & 0.000 & 160/160 & 160/160 & 160/160 \\
6 & 2.603 & 0.183 & 160/160 & 156/160 & 154/160 \\
7 & 5.373 & 0.118 & 160/160 & 147/160 & 146/160 \\
8 & 7.238 & 0.601 & 160/160 & 140/160 & 140/160 \\
\bottomrule
\end{tabular}

\caption{QW4 mean gap (\%) by degree bound, followed by counts of primary-product agreement with the continuous reference. Each size contributes 160 bounded-degree cases; agreement counts refer to cost, edge set and accepted sequence, respectively.}
\label{c:agreement}
\end{table}

The ensemble contained 7,327 equal-cost pairs with unequal exact fourth-order coefficients and 368 exact coefficient degeneracies. Among the nondegenerate pairs, no primary product/continuous inversion was found. Nevertheless, the continuous eight-vertex reference inverted two pair signs relative to QW4 at the fixed time. Both were confirmed in the earlier 70-digit audit and neither altered the returned tree. Equality of the 640 output trees therefore does not establish equality of all edge priorities or a uniform time window.

The 40 differing product/continuous sequences were examined after the run, without changing the policy. Every first divergent accepted edge involved equal original costs and exactly equal fourth-order coefficients. Thirty-seven sequences changed the final tree while retaining cost. These decisions lie outside the nondegeneracy assumption of the proposition. The primary products preserved the fourth-order probability coefficient with maximum numerical discrepancy $1.81\times10^{-16}$, but that local coefficient condition alone does not fix higher-order decisions at finite time.

All seven rules had optimal unrestricted cost on all 320 inputs, giving 2,240 verified trees. This follows from their common cost-first Kruskal structure and cannot establish the quality of a secondary score. In the 160 bounded-degree cases with continuous, distinct costs, all seven rules had identical sequences: the secondary key was unused. The cost comparisons below include these cases and report the complete fixed ensemble.

\subsection{Constrained cost quality}

Table~\mbox{2} gives all seven rules, including QW4 and its two coefficient ablations. QW4 reached 542 optima among 640 cases, with a pooled mean gap of 2.196\%. The full and primary-product references have the same gaps because their costs agree with QW4. The strength rule had gap 1.184\%, the one--three ablation 1.219\%, stable order 3.121\% and one random order 2.648\%. Keeping both interference contributions is necessary to recover the transition coefficient; it does not make that coefficient the best cost heuristic on this ensemble.

The paired QW4-minus-stable gap difference was $-0.925$ percentage points, with interval $[-2.195,0.170]$, from 19 cost wins, 18 losses and 603 ties. Against strength, it was $+1.012$ points with interval $[0.445,1.705]$, from three wins, 31 losses and 606 ties. Thus the observed mean improvement over stable order has an interval crossing zero, whereas the strength comparator had lower mean cost in this ensemble. Against random-one the difference was $-0.452$ points, interval $[-1.388,0.336]$. Against the larger-budget best-of-30 it was $+1.383$ points, interval $[0.782,2.112]$, with zero wins, 39 losses and 601 ties. Figure~\ref{c:benchmarks}(a) displays these predeclared paired comparisons.

The sizes are deliberately small enough to permit independent optimum checks. These results establish a computable classical reduction and delimit its quality on the stated families; they do not establish general optimization superiority or empirical runtime scaling. The strength control also has cheaper score preprocessing: row sums take $O(M+N)$ on an edge list, before the common sorting and forest costs.

\subsection{Descriptive subset with quantized costs and \texorpdfstring{$\Delta<N-1$}{Delta < N - 1}}
\label{c:subset420}

The 640-case aggregate includes 160 cases with distinct continuous costs, for which every secondary rule gives the same construction, and 80 four-vertex cases at $\Delta=3$, for which cost is unrestricted. Their overlap contains 20 cases. To expose comparisons outside these two nondiscriminating categories, Table~\ref{c:quality420} summarizes the fixed subset of $640-(160+80-20)=420$ quantized-cost cases with $\Delta<N-1$. It contains 60, 120, 120 and 120 cases at $N=4,6,7,8$, respectively. This selection uses the input definitions and degree bounds, not observed gaps or wins. Excluded continuous-cost cases need not attain the constrained optimum; their costs simply cannot distinguish the tie breakers.

\begin{table}[htbp]
\centering\small

\begin{tabular}{lrr}
\toprule
Method & Mean gap (\%) & Optima/420 \\
\midrule
Stable & 4.236 & 356 \\
Vertex strength & 1.283 & 377 \\
Two-step term & 7.588 & 321 \\
One--three term & 1.337 & 377 \\
QW4 & 2.826 & 352 \\
Random (one) & 3.515 & 356 \\
Random (best of 30) & 0.719 & 390 \\
\bottomrule
\end{tabular}

\caption{Post-review descriptive summary of 420 cases from the unchanged validation ensemble: quantized costs and $\Delta<N-1$. Gaps use the same enumerated optima as the full benchmark. All seven methods and their original budgets are retained; no graphs, trees, resampling or significance tests were added.}
\label{c:quality420}
\end{table}

QW4 attains 352 optima with mean gap 2.826\% in this subset, compared with 377 optima and gap 1.283\% for strength. The one--three ablation gives 377 optima and gap 1.337\%; best-of-30 gives 390 and 0.719\% using its larger search budget. These descriptive comparisons retain the stronger controls. They do not replace the full 640-case results or introduce new inferential intervals.

\subsection{Finite samples distinguish cost agreement from mechanism resolution}
\label{c:sampling}

The sealed ideal native circuits were additionally sampled at 8,192 and 131,072 shots per true source, with 100 independent single-tree realizations per cell. Every realization produced its own tree; no best-of-100 selection was made. The estimator ordered tied edges by the integer count sum $K_{i|j}+K_{j|i}$, with residual index ties. For six and seven vertices, all eight output categories were retained and no unused outcome was removed from the denominator.

For each graph $G$, the registered pair set $\mathcal P_G$ contains all pairs of edges $e,f$ in the same original cost class with unequal exact rational fourth-order coefficients. Its contrast is $z_{ef}=s_e-s_f$, with $\widehat s_{ij}=(K_{i|j}+K_{j|i})/(2S)$ at $S$ shots per source. Distinct continuous costs supply no eligible pairs, but their graphs remain in every tree-cost metric.

Marginal Clopper--Pearson intervals~\cite{Clopper1934} $[l_{ij},u_{ij}]$ with Bonferroni over $N(N-1)$ active off-diagonal entries supplied nominal simultaneous 95\% coverage per probability table. Each tail uses $0.05/[2N(N-1)]$. Edge-score bounds are $L_e=(l_{ij}+l_{ji})/2$ and $U_e=(u_{ij}+u_{ji})/2$, so the contrast interval is $[L_e-U_f,U_e-L_f]$. Multinomial columns were drawn jointly, and no independence between destination counts was assumed.

The reference contrast uses the complete continuous evolution at the registered time. Its sign is treated as numerically resolved only when $|z_{ef}^{\mathrm{ref}}|>2\times10^{-10}\tau^4$. A contrast is correctly resolved if its interval excludes zero in that direction, wrongly resolved if it excludes zero in the opposite direction, and ambiguous otherwise. Numerically unresolved references are also classified as ambiguous. For each size, method and shot budget, the trial denominator is $100\sum_G|\mathcal P_G|$, including these ambiguous trials. It is 110,400 at six vertices; it counts edge-pair comparisons across the 100 realizations, rather than independent graphs or trees.

At the largest budget, four-vertex direct synthesis gave 88.59\% correct raw signs but only 34.94\% confidently resolved correct contrasts. The primary product gave 34.24\% resolved correct contrasts. At six vertices, direct and product resolved one and three comparisons, respectively, out of 110,400 contrast trials including repetitions. Neither method resolved any at seven or eight vertices; no confidently wrong contrasts were observed. These are results at fixed budgets on the sampled inputs, not an impossibility theorem. Figure~\ref{c:benchmarks}(b) shows the resolution fractions.

\begin{figure}[htbp]
\centering
\includegraphics[width=\linewidth]{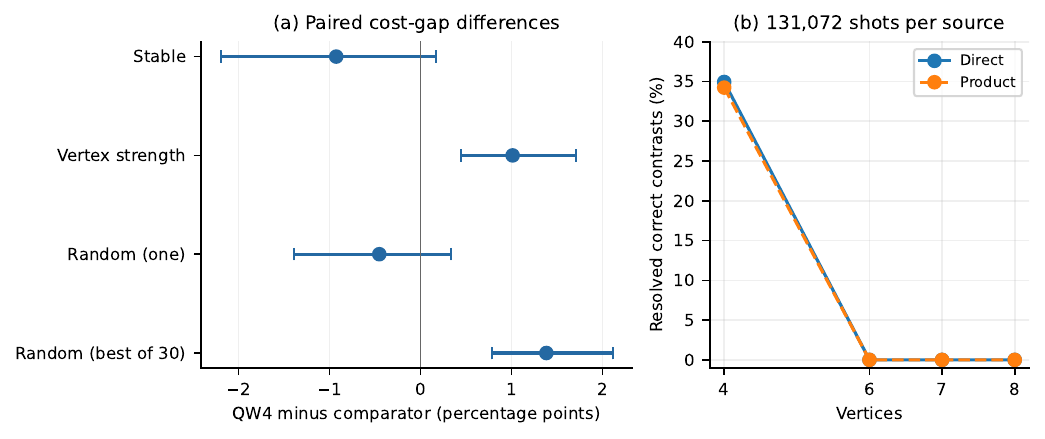}
\caption{Independent validation. (a) Paired mean gap differences with registered 95\% latent-base bootstrap intervals; positive values mean higher cost for QW4. (b) Confidently resolved correct contrasts from ideal compiled probabilities at 131,072 shots per source. Evolution times are $\tau=1$ for $N=4$ and $\tau\simeq0.292864$ for $N=6,7,8$: size and time vary jointly. The six-vertex fractions are one and three of 110,400; both larger sizes have zero. These samples are classical simulations of ideal circuit measurements, not new physical data}
\label{c:benchmarks}
\end{figure}

The corresponding direct-synthesis mean gaps were 0.539\%, 2.268\%, 4.035\% and 4.699\%, compared with 0.389\%, 2.230\%, 3.900\% and 4.917\% for one independently randomized tie order per realization. Cost agreement can remain high even when the interference contrasts are unresolved. This is consistent with a primary cost key that is preserved while finite samples perturb only tied priorities.

The earlier independent audit reconstructed 960 optima, checked 11,040 ideal tree records, and compared 185 selected cells with 70-digit arithmetic. Native-circuit audits checked all 1,440 candidates; sampled-tree audits checked 256,000 bounded-degree trees and 64,000 random-order trees. The present reduction reconstructed its sequences, costs, exact coefficients and paired summaries from these preserved records. No new validation inputs or quantum jobs were used in this reconstruction.

\FloatBarrier

\section{Entropy, conditioning and extensions}
\label{sm:entropy}

We consider entropy for a specified channel and conditioning. In the present known-source unitary protocol, $\rho_j(\tau)=U\ket{j}\bra{j}U^\dagger$ is pure and has zero von Neumann entropy, $S(\rho)=-\tr(\rho\log_2\rho)$. Complete dephasing in the vertex basis instead gives
\begin{equation}
 S\!\left(\operatorname{diag}\rho_j(\tau)\right)
 =-\sum_iP_{ij}(\tau)\log_2P_{ij}(\tau),
 \label{d:entropy}
\end{equation}
with $0\log 0=0$. This entropy characterizes the distribution after that channel. The arithmetic mean of these conditioned entropies is a different functional from the entropy of the uniformly averaged state, $N^{-1}\sum_j\rho_j=I/N$ on the vertex subspace. If an extension assigns different Hamiltonians to candidate trees, their entropies must be compared using a common input scale. These distinctions make entropy suitable for a separate dephasing analysis; the current controller is defined by costs and transition coefficients and has no entropy-maximization guarantee for the constrained objective.

A degree-dependent Hamiltonian or channel would change the dynamical model. It would require a new preparation and evolution specification, a new probability expansion and an independent assessment of the resulting controller. The present fixed-generator analysis establishes the reference against which such extensions could be compared. Its conclusions concern the static protocol and the stated measurement and validation records.

\bibliographystyle{sn-vancouver-num}
\bibliography{references}